\documentclass[12pt,unpublished, onecolumn,a4paper]{quantumarticle}

\pdfoutput=1 
\usepackage[T1]{fontenc}
\input{Qcircuit} 
\usepackage{makecell}
\usepackage{tikz}
\usetikzlibrary{patterns}
\usetikzlibrary{shapes}
\usepackage{comment}
\usetikzlibrary{positioning, calc}
\usetikzlibrary{calc}
\usetikzlibrary{arrows}
\usepackage{tikz-3dplot}
\usepackage{graphicx}
\usepackage{mdwlist}
\usepackage{color}
\usepackage{thmtools}
\usepackage{amssymb}
\usepackage{physics}
\usepackage{amsmath}
\usepackage{array}
\usepackage{doi}
\usepackage{url,hyperref}
\usepackage[capitalize,nameinlink]{cleveref}
\hypersetup{colorlinks={true},linkcolor={blue},citecolor=red}
\usepackage[table,xcdraw]{xcolor}

\newcommand{\uma}[1]{\textcolor{red}{Uma:#1}}

\usepackage[numbers,sort&compress]{natbib}
\usetikzlibrary{fadings}
\usetikzlibrary{decorations.pathmorphing}
\usepackage{mathrsfs}
\usepackage{authblk}
\usepackage[mathscr]{euscript}
\usepackage{enumitem}

\newcommand{\vabs}[1]{\left\Vert #1 \right\Vert}

\DeclareMathAlphabet\mathbfcal{OMS}{cmsy}{b}{n}

\usepackage{cleveref}

\usetikzlibrary{decorations.pathreplacing}
\tikzset{snake it/.style={decorate, decoration=snake}}

\usetikzlibrary{decorations.pathreplacing,decorations.markings}

\tikzset{
    >=stealth',
    punkt/.style={
           rectangle,
           rounded corners,
           draw=black, very thick,
           text width=6.5em,
           minimum height=2em,
           text centered},
    pil/.style={
           ->,
           thick,
           shorten <=2pt,
           shorten >=2pt,},
  on each segment/.style={
    decorate,
    decoration={
      show path construction,
      moveto code={},
      lineto code={
        \path [#1]
        (\tikzinputsegmentfirst) -- (\tikzinputsegmentlast);
      },
      curveto code={
        \path [#1] (\tikzinputsegmentfirst)
        .. controls
        (\tikzinputsegmentsupporta) and (\tikzinputsegmentsupportb)
        ..
        (\tikzinputsegmentlast);
      },
      closepath code={
        \path [#1]
        (\tikzinputsegmentfirst) -- (\tikzinputsegmentlast);
      },
    },
  },
  mid arrow/.style={postaction={decorate,decoration={
        markings,
        mark=at position .5 with {\arrow[#1]{stealth'}}
      }}}
}

\usepackage{hyperref}
\usepackage{slashed}
\usepackage{float} 
\usepackage{graphicx} 
\usepackage{subcaption}

\usepackage{bbold}

\mathchardef\mhyphen="2D
\newcommand{\forr}{\mathrm{forr}}

\newcommand{\PDT}{\mathsf{PDT}}
\newcommand{\RPDT}{\mathsf{RPDT}}

\newcommand{\PSM}{\mathsf{PSM}}
\newcommand{\PSQM}{\mathsf{PSQM}}

\newcommand{\na}{\mathrm{na}}

\newcommand{\pub}{\mathsf{pub}}
\newcommand{\R}{\mathsf{R}}
\newcommand{\Rent}{\mathsf{R}\|^*}
\newcommand{\Rsim}{\mathsf{R}\|}
\newcommand{\Qent}
{\mathsf{Q}\|^*}
\newcommand{\Qsim}{\mathsf{Q}\|}
\newcommand{\Dsim}{\mathsf{D}\|}
\newcommand{\Qone}{\mathsf{Q}1}
\newcommand{\Rone}{\mathsf{R}1}
\newcommand{\Qtwo}{\mathsf{Q}2}
\newcommand{\Rtwo}{\mathsf{R}}

\newtheorem{theorem}{Theorem}

\newtheorem{corollary}[theorem]{Corollary}

\newtheorem{definition}[theorem]{Definition}

\newtheorem{lemma}[theorem]{Lemma}

\newtheorem{remark}[theorem]{Remark}

\newenvironment{proof}[1][Proof]{\noindent\textbf{#1. }}{\ \rule{0.5em}{0.5em}}
\newcommand{\am}[1]{\textcolor{blue}{#1}}

\begin{document}

\title{Magic and communication complexity}

\author[1]{Uma Girish}
\email{ug2150@columbia.edu}
\orcid{0000-0003-3055-9406}

\author[2,3]{Alex May}
\email{amay@perimeterinstitute.ca}
\orcid{0000-0002-4030-5410}

\author[1]{Natalie Parham}
\email{natalie@cs.columbia.edu}
\orcid{}

\author[1]{Henry Yuen}
\email{hyuen@cs.columbia.edu}
\orcid{0000-0002-2684-1129}

\affiliation[1]{Columbia University}
\affiliation[2]{Perimeter Institute for Theoretical Physics}
\affiliation[3]{Institute for Quantum Computing, University of Waterloo}

\abstract{
We establish novel connections between magic in quantum circuits and communication complexity. In particular, we show that functions computable with low magic have low communication cost.

Our first result shows that the $\Dsim$ (deterministic simultaneous message passing) cost of a Boolean function $f$ is at most the number of single-qubit magic gates in a quantum circuit computing $f$ with any quantum advice state. If we allow mid-circuit measurements and adaptive circuits, we obtain an upper bound on the two-way communication complexity of $f$ in terms of the magic + measurement cost of the circuit for $f$. 
As an application, we obtain magic-count lower bounds of $\Omega(n)$ for the $n$-qubit generalized Toffoli gate as well as the $n$-qubit quantum multiplexer.


Our second result gives a general method to transform $\Qent$ protocols (simultaneous quantum messages with shared entanglement) into $\Rent$ protocols (simultaneous classical messages with shared entanglement) which incurs only a polynomial blowup in the communication and entanglement complexity, provided the referee's action in the $\Qent$ protocol is implementable in constant $T$-depth. The resulting $\Rent$ protocols satisfy strong privacy constraints and are $\PSM^*$ protocols (private simultaneous message passing with shared entanglement), where the referee learns almost nothing about the inputs other than the function value. As an application, we demonstrate $n$-bit partial Boolean functions whose $\Rent$ complexity is $\mathrm{polylog}(n)$ and whose $\R$ (interactive randomized) complexity is $n^{\Omega(1)}$, establishing the first exponential separations between $\Rent$ and $\R$ for Boolean functions.  
}

\maketitle


\tableofcontents


\section{Introduction and summary}

We explore the connection between magic (non-Clifford) gates in quantum computation and communication complexity. 
A central result of this work is that functions computable with low magic are also easy from the perspective of communication complexity.
In particular, we find that in several models a small magic gate count leads to a small communication cost.
Several variations on this theme occur, with low magic count in differing models of computation leading to low communication cost in a corresponding model of communication. 
As another observation at the interface of magic and communication, we find that if a communication protocol itself uses low-magic computations, it can in some contexts be transformed to use weaker resources. 
In particular, with the use of entanglement, a simultaneous quantum message passing protocol in which the referee uses low magic operations can efficiently be made to use only classical messages, and furthermore made to have a strong privacy property. 

Magic gates play a special role in quantum computation. 
In particular, Clifford$+T$ is the most widely considered gate set, and is the basis for many fault-tolerant quantum computation schemes. 
Typically in these schemes Clifford gates are implemented directly by acting on the encoded qubits, while the $T$ gates are implemented by preparing and injecting magic states, see e.g. \cite{bravyi2005universal}. 
This distinction makes $T$ gates particularly costly, leading to interest in understanding how many of them are really necessary for a given computation. 
As well, circuits with low numbers of magic gates can be efficiently simulated by a classical computer \cite{gottesman1998heisenberg,aaronson2004improved}, highlighting the role of magic gates in quantum advantage.

Communication complexity is an important tool for lower bounding classical computational models. 
For instance, classical two-way communication complexity is a lower bound on decision tree complexity \cite{nisan1993communication}. 
Another example is the Karchmer-Widgerson technique which relates circuit lower bounds for a function $f$ to the communication complexity of a relation determined by $f$ \cite{karchmer1988monotone}.\footnote{See \cite[Chapters 9-11]{kushilevitz1997communication} for a review of classical complexity lower bounds based on communication complexity.}
Given this, it is natural to ask if communication complexity can also be used to lower bound quantum computational complexity. 
Our lower bounds begin to address this question. 

One antecedent to our work occurs in the context of non-local quantum computation (NLQC) \cite{speelman2015instantaneous}, where it was observed that unitaries with low $T$-depth can be implemented efficiently. 
The NLQC setting involves a single simultaneous round of communication, and asks for the application of a unitary, so is somewhat different than more standard communication complexity settings. 
Nonetheless we find techniques from this earlier result to be useful in our context, where we upper bound communication cost in terms of $T$-depth or magic gate count.

\textbf{Note:} During the preparation of this manuscript we realized similar results to two of our magic gate lower bounds (proven in \Cref{sec:magiclowerbounds}) had been proven concurrently by another group, see \cite{gosset2025multi}. 
They use a different technique to obtain similar bounds on the $T$-count in the unitary and mixed settings.
We discuss the relationship of our results at the end of \Cref{sec:magiclowerbounds}.

\subsection{Communication \& magic gate count}


We show that in several settings, Boolean functions computed by low magic quantum circuits have small communication complexity. 
In this section, we outline the proof of the simplest such relationship, and then state some variations.

Consider a (partial or total\footnote{A total Boolean function is defined on all points in $\{0,1\}^n$, while a partial Boolean function may be defined only a subset of inputs and we only care about computing the function on the support.}) Boolean function $f:\{0,1\}^n\to \{0,1\}$. The simplest computational setting to consider is the unitary model, where we consider a quantum circuit $C$ that takes as input $\ket{z}$ in the computational basis for $z\in\{0,1\}^n$, and finally measures the first qubit in the computational basis to obtain the output. 
We also allow the circuit to take in an additional quantum state $\ket{\psi}$, which can be arbitrarily large and can begin in an arbitrary state.
We call this the advice state. 
We say the circuit $C$ computes $f$ with error $\epsilon$ if the measurement outcome yields $f(z)$ with probability at least $1-\epsilon$ for $\epsilon<1/2$. 
We will establish a lower bound on the number of magic gates needed in this model in terms of the deterministic simultaneous message passing communication cost, which we label the $\Dsim$ model. 
See \Cref{fig:SMP}.

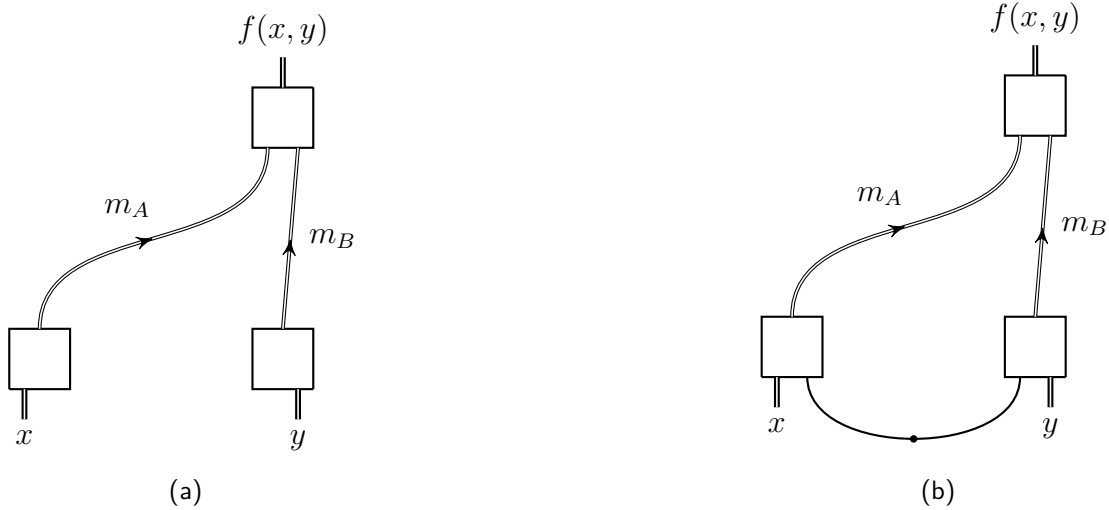
\begin{figure*}
    \centering
    \begin{subfigure}{0.45\textwidth}
    \centering
    \begin{tikzpicture}[scale=0.4]
    
    \draw[thick] (-5,-5) -- (-5,-3) -- (-3,-3) -- (-3,-5) -- (-5,-5);
    
    \draw[thick] (5,-5) -- (5,-3) -- (3,-3) -- (3,-5) -- (5,-5);
    
    \draw[thick] (5,5) -- (5,3) -- (3,3) -- (3,5) -- (5,5);
    
    \draw[double, mid arrow] (4,-3) -- (4.5,3);
    
    \draw[double, mid arrow] (-4,-3) to [out=90,in=-90] (3.5,3);
    
    \draw[double, thick] (-4.5,-6) -- (-4.5,-5);
    \node[below] at (-4.5,-6) {$x$};
    
    \draw[double, thick] (4.5,-6) -- (4.5,-5);
    \node[below] at (4.5,-6) {$y$};

    \node[left] at (0,1) {$m_A$};
    \node[right] at (4.5,0) {$m_B$};
    
    \draw[double, thick] (4,5) -- (4,6);
    \node[above] at (4,6) {$f(x,y)$};
    
    \end{tikzpicture}
    \caption{}
    \label{fig:SMP}
    \end{subfigure}
    \hfill
    \begin{subfigure}{0.45\textwidth}
    \centering
    \begin{tikzpicture}[scale=0.4]
    
    \draw[thick] (-5,-5) -- (-5,-3) -- (-3,-3) -- (-3,-5) -- (-5,-5);
    
    \draw[thick] (5,-5) -- (5,-3) -- (3,-3) -- (3,-5) -- (5,-5);
    
    \draw[thick] (5,5) -- (5,3) -- (3,3) -- (3,5) -- (5,5);
    
    \draw[double, mid arrow] (4,-3) -- (4.5,3);
    
    \draw[double, mid arrow] (-4,-3) to [out=90,in=-90] (3.5,3);
    
    \draw[thick] (-3.5,-5) to [out=-90,in=-90] (3.5,-5);
    \draw[black] plot [mark=*, mark size=3] coordinates{(0,-7.05)};

    \node[left] at (0,1) {$m_A$};
    \node[right] at (4.5,0) {$m_B$};
    
    \draw[double, thick] (-4.5,-6) -- (-4.5,-5);
    \node[below] at (-4.5,-6) {$x$};
    
    \draw[double, thick] (4.5,-6) -- (4.5,-5);
    \node[below] at (4.5,-6) {$y$};
    
    \draw[double, thick] (4,5) -- (4,6);
    \node[above] at (4,6) {$f(x,y)$};
    
    \end{tikzpicture}
    \caption{}
    \label{fig:PSM*}
    \end{subfigure}
    \caption{a) The simultaneous message passing ($\Dsim$) setting. Alice receives input $x\in\{0,1\}^n$, Bob receives input $y\in\{0,1\}^n$, and the referee should output $f(x,y)$. Alice and Bob can not communicate with one another, but can each send a message to the referee. The $\Dsim$ cost is the minimal number of bits of communication Alice and Bob must send. b) The $\PSM^*$ model, which has the same communication pattern as $\Dsim$. In $\PSM^*$, Alice and Bob may share entanglement, which we indicate with the lower curved wire. We restrict the communication to be classical, which is indicated by the double-lined wires. Further, the messages are required to be \emph{private}, meaning that the referee should learn $f(x,y)$ but no other information about $(x,y)$.} 
    \label{fig:SMPandPSM}
\end{figure*}

A useful starting point is to begin with the case where the circuit is Clifford.
Consider an arbitrary division of the input $z$ into $(x,y)$, and ask about the $\Dsim$ cost of computing $f(x,y)$. 
Label the circuit computing $f$ by $C_{ABE}$, where $A$ is a system holding Alice's input, $B$ is a system holding Bob's input, and $E$ is an advice system which may be prepared in an arbitrary state.
We consider having the referee run this circuit on the all-zeroes input, which we can view as
\begin{align}
    C_{ABE}\ket{0}_A\ket{0}_B\ket{\psi}_E &= C_{ABE}X^{\vec{x}}\ket{x}_AX^{\vec{y}}\ket{y}_B\ket{\psi}_E \nonumber \\
    &= \sigma_{ABE}[x,y] C_{ABE}\ket{x}_A\ket{y}_B\ket{\psi}_E.
\end{align}
In the last equality, we have conjugated the Pauli string that sets $\ket{x}\ket{y}$ back to the all-zeroes state through the Clifford, returning a Pauli string acting on the outputs of the circuit. 
In this viewpoint, the circuit is running on the correct inputs. 

The key observation is that to learn $f(x,y)$, we only need to undo the Pauli acting on the first qubit, which is the only one which will be measured.
In fact, a possible $Z$ correction will not disturb the measurement outcome and can be left uncorrected, so we only need to determine the single bit which controls a possible $X$ correction on the output qubit. 
This bit is determined by the parity of a subset of the input bits: on each input the Pauli $X$ conjugates through to the outputs in some way, sometimes giving an $X$ correction on the output qubit. 
Thus $f(x,y)$ is determined by a single parity function $p(x,y)=\sum_{i\in S_A}x_i+ \sum_{i\in S_B}y_i$.
The referee can compute this parity function by having Alice and Bob send the single bits $\sum_{i\in S_A}x_i$ and $\sum_{i\in S_B}y_i$ respectively, so the Clifford case has $\Dsim$ cost of $2$.\footnote{In fact, this observation is already made (up to small differences) in \cite{buhrman2006new}. Our contribution is to extend a similar strategy to the case with $T$-gates.}

We can generalize this strategy to the case of circuits that use magic gates, at the cost of adding a constant number of parity functions, and hence constant $\Dsim$ communication cost, per magic gate. 
With magic gates present, we can compute $f(x,y)$ similarly to before: run the circuit on the all-zeroes input, and conjugate the string of $X$ corrections that would make this the correct input through to the first magic gate appearing in the circuit. 
Undo the Pauli corrections before the first magic gate. 
The identity of the Pauli correction depends on $2c_M$ parities of the input, where $c_M$ is the number of qubits on which the magic gate acts. 
Continue in this way, correcting the Paulis only on the wires before each magic gate until reaching the measurement, which requires $1$ additional parity value. 
Thus we obtain an upper bound of $4c_M\cdot \mathcal{M}_{\epsilon,c_M}^{\text{unitary}}(f) + 2$ on the $\Dsim$ complexity, where $\mathcal{M}_{\epsilon,c_M}^{\text{unitary}}(f)$ is the number of magic gates in any circuit that computes $f$ with probability $1-\epsilon$, where the magic gates act on at most $c_M$ qubits. 
Thus we obtain the lower bound
\begin{align}
    \frac{1}{4c_M}\left(\Dsim(f)-2\right) \leq \mathcal{M}^{\text{unitary}}_{\epsilon,c_M}(f)\quad\text{for all }\epsilon<1/2.
\end{align}
This implies that the number of magic gates needed to compute $f$ with any error $\epsilon<1/2$ is essentially lower bounded by the $\Dsim$ communication complexity.
See \Cref{thm:unitarymagiclowerbound} in the main text.

We can also improve this bound for the $T$-gate count: in this case $c_M=1$, but we also notice that since $T$ gates commute with Pauli $Z$, we can leave $Z$'s uncorrected as we move through the circuit. Thus, we have
\begin{align}
    \frac{1}{2}(\Dsim(f)-2) \leq \mathcal{T}_{\epsilon<1/2}^{\text{unitary}}(f).
\end{align}

We can use a similar technique to upper bound the communication cost in terms of magic gate count in other computational models. 
One modification of the above is to consider the \emph{mixed unitary model}, where we allow quantum operations of the form
\begin{align}
    \mathcal{N}(\cdot) = \sum_i p_i U_i(\cdot) U^\dagger_i.
\end{align} 
We say that the above channel computes $f$ if measuring the first qubit yields $f(z)$ with probability at least $1-\epsilon$. 
Let $\mathcal{M}^{\text{mixed}}_{\epsilon,c_M}(f)$ denote the maximum number of magic gates used by the $U_i$, where each magic gate acts on at most $c_M$ qubits. 
Then we obtain that
\begin{align}
    \frac{1}{4c_M}\left(\Rsim^{\pub}_\epsilon(f)-2 \right) &\leq \mathcal{M}^{\text{mixed}}_{\epsilon,c_M}(f),
\end{align}
where the communication model $\Rsim^{\pub}_\epsilon$ now allows public randomness, and has the same error probability $\epsilon$ as the circuit. 
See \Cref{thm:mixedlowerbound} in the main text. 

A second variation is to consider an adaptive Clifford+magic gate model. 
In this model, we allow an arbitrary advice state as before, as well as mid-circuit measurements. 
Further gates can then be applied adaptively, where we condition on mid-circuit measurement outcomes. 
This model allows, for instance, the use of magic state injection and mimics the model expected to be implemented in a fault tolerant quantum computer. 
In this setting we take the cost to be the number of magic gates plus the number of mid-circuit measurements.\footnote{Note that because we allow an arbitrary advice state, which could include magic states, we cannot hope to obtain a lower bound purely in terms of the magic gate count alone.} 
We denote this cost by $\mathcal{M}^{\text{adaptive}}_{\epsilon,c_M}(f)$. 
We obtain the following lower bound
\begin{align}
    \frac{1}{2c_M}(\R_\epsilon(f)-1) \leq \mathcal{M}_{\epsilon,c_M}^{\text{adaptive}}(f).
\end{align}
Here $\R_\epsilon(f)$ denotes the two-way classical communication complexity of computing $f$ with probability $1-\epsilon$.
See \Cref{thm:adaptivelowerbound} in the main text. 

\paragraph*{Applications:} Because well developed lower bound strategies are known for classical communication complexity, we obtain bounds on the magic gate complexity of many explicit Boolean functions. 
Somewhat less directly, we can also use our Boolean function lower bounds to bound the magic gate complexity of unitaries. 
To do this, the strategy is to find Boolean functions with large communication complexity that are computed (with small magic overhead) by the unitary of interest. 

One unitary of interest is the $n$-qubit Toffoli, which acts according to
\[
    \mathrm{Toffoli}_n : \ket{x_1,\ldots,x_n,b} \mapsto \ket{x_1,\ldots,x_n,b \oplus \bigwedge_{i=1}^n x_i}~.
\]
This can be used to compute the equality function with no magic overhead: we take the bit-wise XOR of the input strings, and negate every bit of output. 
The result is the all 1's string iff the input strings are equal, which we check using the generalized Toffoli. 
The $\Dsim$ complexity of equality is $n$, so this gives a $\Omega(n)$ lower bound on the magic gate count of the Toffoli in the exact case. 
Considering implementing Toffoli aproximately in the mixed model, the relevant communication lower bound is a $\Omega(\min\{n,1/\epsilon\})$ lower bound on $\Rsim$, leading to the same lower bound on the magic gate count in that case. 
The concurrent work \cite{gosset2025multi} also find these lower bounds, and in fact gives nearly matching upper bounds. 

Another unitary of interest is the quantum multiplexer, which acts according to
\[
    \mathrm{Multiplex}_n: \ket{i,x,b}  \mapsto \ket{i,x_1,\ldots,x_{i-1},b,x_{i+1},\ldots,x_n,x_i}~.
\]
In words, the quantum multiplexer coherently swaps the bit $b$ into the register labelled by $i$. 
The quantum multiplexer can be used, with zero magic overhead, to compute the index function, Index$_n(x,i)=x_i$.
Since Index$_n$ has a $\Omega(n)$ lower bound in the $\Rsim$ model, this gives an $\Omega(n)$ lower bound on the magic gate count to implement the quantum multiplexer in the mixed model.

\subsection{Communication \& magic depth}


In this part, we study simultaneous message passing protocols with constant error $\epsilon=1/3$, where Alice and Bob share entanglement. We consider two models, namely $\Qent$, where the messages to the referee are quantum and $\Rent$, where the messages are classical. Building on techniques from the non-local quantum computation literature~\cite{speelman2015instantaneous}, we give a general technique to convert $\Qent$ protocols into $\Rent$ protocols.
The cost of the conversion is determined by the complexity of the referee's action. 
In particular, if the $T$-depth of the referee's actions in the $\Qent$ protocol is $O(1)$, then the conversion only incurs a polynomial overhead in the entanglement and communication complexity. More formally, if the referee receives $m$ qubits from Alice and Bob, uses $a$ ancillary qubits of quantum advice, implements a unitary of $T$-depth $d$, and finally measures the first qubit to obtain the output, we show that
\begin{align}
    \Rent(f)\le (O(m+a))^d.
\end{align}
Additionally, our $\Rent$ protocol has strong privacy conditions and is in fact also a $\PSM^*$ protocol, where the referee learns almost nothing about Alice's and Bob's inputs except for the output of the function.
See figure \ref{fig:PSM*} for an illustration of the $\PSM^*$ model.
See~\Cref{thm:Q||toPSM} for a formal statement of our transformation.  

The main idea behind the proof of~\Cref{thm:Q||toPSM} is to 
apply a technique from \cite{speelman2015instantaneous}. 
Heuristically, we have Alice and Bob in the $\PSM^*$ protocol themselves (nearly) implement what was previously the referee's operation in the $\Qent$ protocol, and leave only certain simple correction operations to be performed by the referee. 
The data needed for these simple correction operations turn out to naturally be classical bits, and to reveal only the function value. 

In more detail, we build a $\PSM^*$ protocol from a $\Qent$ protocol as follows. 
Alice and Bob first implement the operations from the $\Qent$ protocol to produce the message systems. 
Then, Bob teleports his message system to Alice, who will attempt to implement the referee's actions. 
If the referee's actions were Clifford, this would be simple: Alice applies the needed Clifford, measures the output qubit, and produces a measurement outcome which she sends to the referee. 
Using this measurement outcome and the teleportation outcomes from Bob, the referee can determine $f(x,y)$. 
With $T$ gates present the situation is more complicated. 
However, \cite{speelman2015instantaneous} shows how to use repeated teleportations between Alice and Bob to have Alice apply $T$ gates instantaneously, before the communication, at least up to Pauli corrections. 
The Pauli corrections are determined as a function of all of the teleportation measurement outcomes. 
Further, the number of these measurement outcomes that must be communicated grows in a controlled way as the $T$-depth of the circuit applied by Alice increases. 
Thus if Alice makes the needed measurement and both Alice and Bob communicate their teleportation measurement outcomes, then the referee can determine $f(x,y)$, giving a $\Rent$ protocol. 
It is not too difficult to show that these messages store (almost) no data about $(x,y)$ except the value of $f(x,y)$. 
Indeed, teleportation measurement outcomes are uniformly random and reveal nothing about the input. 
The only step that reveals information is Alice's bit, but since she only reveals a single bit that correlates highly with $f(x,y)$, she doesn't end up revealing more information. 
For more details on the proof, see~\Cref{sec:Q||toPSM}. 

\paragraph*{Applications.}
Our result has applications to quantum speedups in communication complexity. 
We begin by providing a brief overview of this field and motivating our results in this context. 

The study of quantum speedups in communication complexity has a long and rich history. Numerous works~\cite{raz1999exponential,gavinskyetal,klartagregev,gavinsky2016entangled,girish2022quantum,arunachalam2023one} have shown partial Boolean functions for which quantum communication provides exponential speedups over classical communication. 
Each subsequent work either strengthens the classical lower bound or weakens the resources required by the quantum protocol. 
The best known prior separations for partial Boolean functions are due to~\cite{gavinsky2016entangled,girish2022quantum,arunachalam2023one} which prove that $\Qent$ can exponentially outperform $\R$ (randomized interactive communication). 
These are depicted in~\Cref{fig:2}. 

Despite decades of work, our understanding of quantum communication speedups is far from complete (see~\cite{gavinsky2019,barequantum20} for a list of open problems). 
In particular, we paraphrase two open questions proposed by~\cite{gavinsky2016entangled}:
\begin{enumerate}
    \item Is there a partial function separating $\Rent$ and $\R$? 
    \item Is there a partial function or even a relational problem separating $\Qent$ and $\Rent$? 
\end{enumerate}

In this work, we resolve the first question and make some progress on the second, by using our aforementioned connection between communication complexity and magic depth. Firstly,~\Cref{thm:Q||toPSM} implies that separating $\Qent$ and $\Rent$ requires proving $T$-depth lower bounds on the measurement implemented by the referee in the $\Qent$ protocol.
Secondly, we show the first exponential separation between $\Rent$ and $\R$ for partial Boolean functions. 
We do this by taking existing separations between $\Qent$ and $\R$ where the referee's actions have constant $T$-depth and applying our conversion to obtain $\Rent$ protocols. 
In particular, we use a variant of the distributed Forrelation Problem introduced by~\cite{girish2022quantum} and the ABCD problem introduced by~\cite{arunachalam2023one}. 
These works showed that these problems require $\Rtwo$ protocols of cost $\tilde{\Omega}(n^{1/4})$ and this establishes the desired separation between $\Rent$ and $\Rtwo$. 
The best known prior separation between these models was a relational one~\cite{Gav08}.
Compared to previous separations, it seems notable that separating $\Rent$ and $\R$ uses a quite involved upper bound strategy, in particular the technique of \cite{speelman2015instantaneous}.
In contrast, most prior quantum communication upper bounds use simpler strategies.
As mentioned before, our $\Rent$ protocol has the additional advantage that it is also a $\PSM^*$ protocol.

\tikzset{every picture/.style={line width=0.75pt}} 
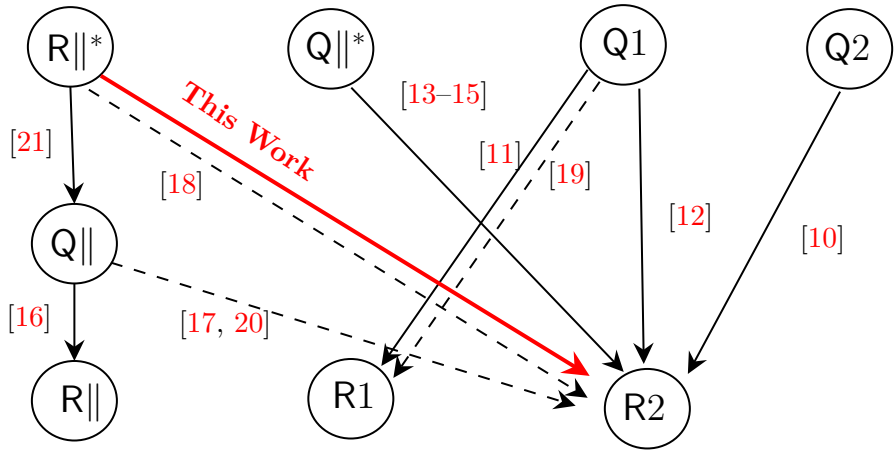
\begin{figure}
    \centering

\tikzset{every picture/.style={line width=0.75pt}} 

\begin{tikzpicture}[x=0.75pt,y=0.75pt,yscale=-1,xscale=1]

\draw   (251,55) .. controls (251,43.95) and (260.5,35) .. (272.22,35) .. controls (283.94,35) and (293.44,43.95) .. (293.44,55) .. controls (293.44,66.05) and (283.94,75) .. (272.22,75) .. controls (260.5,75) and (251,66.05) .. (251,55) -- cycle ;

\draw   (123,153) .. controls (123,141.95) and (132.5,133) .. (144.22,133) .. controls (155.94,133) and (165.44,141.95) .. (165.44,153) .. controls (165.44,164.05) and (155.94,173) .. (144.22,173) .. controls (132.5,173) and (123,164.05) .. (123,153) -- cycle ;

\draw   (121,54) .. controls (121,42.95) and (130.5,34) .. (142.22,34) .. controls (153.94,34) and (163.44,42.95) .. (163.44,54) .. controls (163.44,65.05) and (153.94,74) .. (142.22,74) .. controls (130.5,74) and (121,65.05) .. (121,54) -- cycle ;

\draw   (397,53) .. controls (397,41.95) and (406.5,33) .. (418.22,33) .. controls (429.94,33) and (439.44,41.95) .. (439.44,53) .. controls (439.44,64.05) and (429.94,73) .. (418.22,73) .. controls (406.5,73) and (397,64.05) .. (397,53) -- cycle ;

\draw   (510,55) .. controls (510,43.95) and (519.5,35) .. (531.22,35) .. controls (542.94,35) and (552.44,43.95) .. (552.44,55) .. controls (552.44,66.05) and (542.94,75) .. (531.22,75) .. controls (519.5,75) and (510,66.05) .. (510,55) -- cycle ;

\draw   (261,231) .. controls (261,219.95) and (270.5,211) .. (282.22,211) .. controls (293.94,211) and (303.44,219.95) .. (303.44,231) .. controls (303.44,242.05) and (293.94,251) .. (282.22,251) .. controls (270.5,251) and (261,242.05) .. (261,231) -- cycle ;

\draw   (409,236) .. controls (409,224.95) and (418.5,216) .. (430.22,216) .. controls (441.94,216) and (451.44,224.95) .. (451.44,236) .. controls (451.44,247.05) and (441.94,256) .. (430.22,256) .. controls (418.5,256) and (409,247.05) .. (409,236) -- cycle ;

\draw    (526.44,76.52) -- (451.86,215.88) ;
\draw [shift={(450.44,218.52)}, rotate = 298.16] [fill={rgb, 255:red, 0; green, 0; blue, 0 }  ][line width=0.08]  [draw opacity=0] (10.72,-5.15) -- (0,0) -- (10.72,5.15) -- (7.12,0) -- cycle    ;
\draw    (426.22,75) -- (428.39,210.52) ;
\draw [shift={(428.44,213.52)}, rotate = 269.08] [fill={rgb, 255:red, 0; green, 0; blue, 0 }  ][line width=0.08]  [draw opacity=0] (10.72,-5.15) -- (0,0) -- (10.72,5.15) -- (7.12,0) -- cycle    ;

\draw [color={rgb, 255:red, 0; green, 0; blue, 0 }  ,draw opacity=1 ] [dash pattern={on 4.5pt off 4.5pt}]  (406.44,71.52) -- (305.15,218.05) ;


\draw [color={rgb, 255:red, 0; green, 0; blue, 0 }  ,draw opacity=1 ]   (400.44,65.52) -- (299.15,212.05) ;

\draw [shift={(297.44,214.52)}, rotate = 304.66] [fill={rgb, 255:red, 0; green, 0; blue, 0 }  ,fill opacity=1 ][line width=0.08]  [draw opacity=0] (10.72,-5.15) -- (0,0) -- (10.72,5.15) -- (7.12,0) -- cycle    ;

\draw [shift={(303.44,220.52)}, rotate = 304.66] [fill={rgb, 255:red, 0; green, 0; blue, 0 }  ,fill opacity=1 ][line width=0.08]  [draw opacity=0] (10.72,-5.15) -- (0,0) -- (10.72,5.15) -- (7.12,0) -- cycle    ;
\draw [color={rgb, 255:red, 0; green, 0; blue, 0 }  ,draw opacity=1 ]   (282.44,74.52) -- (416.37,215.35) ;
\draw [shift={(418.44,217.52)}, rotate = 226.44] [fill={rgb, 255:red, 0; green, 0; blue, 0 }  ,fill opacity=1 ][line width=0.08]  [draw opacity=0] (10.72,-5.15) -- (0,0) -- (10.72,5.15) -- (7.12,0) -- cycle    ;
\draw [color={rgb, 255:red, 0; green, 0; blue, 0 }  ,draw opacity=1 ] [dash pattern={on 4.5pt off 4.5pt}]  (151.44,75.52) -- (397.89,228.94) ;
\draw [shift={(400.44,230.52)}, rotate = 211.9] [fill={rgb, 255:red, 0; green, 0; blue, 0 }  ,fill opacity=1 ][line width=0.08]  [draw opacity=0] (10.72,-5.15) -- (0,0) -- (10.72,5.15) -- (7.12,0) -- cycle    ;
\draw [color={rgb, 255:red, 0; green, 0; blue, 0 }  ,draw opacity=1 ] [dash pattern={on 4.5pt off 4.5pt}]  (162.44,162.52) -- (391.58,233.63) ;
\draw [shift={(394.44,234.52)}, rotate = 197.24] [fill={rgb, 255:red, 0; green, 0; blue, 0 }  ,fill opacity=1 ][line width=0.08]  [draw opacity=0] (10.72,-5.15) -- (0,0) -- (10.72,5.15) -- (7.12,0) -- cycle    ;
\draw   (123,232) .. controls (123,220.95) and (132.5,212) .. (144.22,212) .. controls (155.94,212) and (165.44,220.95) .. (165.44,232) .. controls (165.44,243.05) and (155.94,252) .. (144.22,252) .. controls (132.5,252) and (123,243.05) .. (123,232) -- cycle ;
\draw    (144.22,173) -- (144.22,209) ;
\draw [shift={(144.22,212)}, rotate = 270] [fill={rgb, 255:red, 0; green, 0; blue, 0 }  ][line width=0.08]  [draw opacity=0] (10.72,-5.15) -- (0,0) -- (10.72,5.15) -- (7.12,0) -- cycle    ;
\draw    (142.22,74) -- (144.12,130) ;
\draw [shift={(144.22,133)}, rotate = 268.06] [fill={rgb, 255:red, 0; green, 0; blue, 0 }  ][line width=0.08]  [draw opacity=0] (10.72,-5.15) -- (0,0) -- (10.72,5.15) -- (7.12,0) -- cycle    ;
\draw [color={rgb, 255:red, 255; green, 0; blue, 0 }  ,draw opacity=1 ][line width=1.5]    (157.44,68.52) -- (398.04,217.42) ;
\draw [shift={(401.44,219.52)}, rotate = 211.75] [fill={rgb, 255:red, 255; green, 0; blue, 0 }  ,fill opacity=1 ][line width=0.08]  [draw opacity=0] (13.4,-6.43) -- (0,0) -- (13.4,6.44) -- (8.9,0) -- cycle    ;

\draw (128.66,44.97) node [anchor=north west][inner sep=0.75pt]  [font=\large] [align=left] {$\Rent$};
\draw (130.49,143.39) node [anchor=north west][inner sep=0.75pt]  [font=\large] [align=left] {$\Qsim$};
\draw (257.99,44.35) node [anchor=north west][inner sep=0.75pt]  [font=\large] [align=left] {$\Qent$};
\draw (405.94,43.3) node [anchor=north west][inner sep=0.75pt]  [font=\large] [align=left] {$\Qone$};
\draw (271.47,220.93) node [anchor=north west][inner sep=0.75pt]  [font=\large] [align=left] {$\Rone$};
\draw (416.4,226.64) node [anchor=north west][inner sep=0.75pt]  [font=\large] [align=left] {$\Rtwo2$};
\draw (516.81,45.82) node [anchor=north west][inner sep=0.75pt]  [font=\large] [align=left] {$\Qtwo$};
\draw (135.47,221.93) node [anchor=north west][inner sep=0.75pt]  [font=\large] [align=left] {$\Rsim$};
\draw (500,142) node [anchor=north west][inner sep=0.75pt]  [font=\small] [align=left] {~\cite{raz1999exponential}};
\draw (434,131) node [anchor=north west][inner sep=0.75pt]  [font=\small] [align=left] {~\cite{klartagregev}};
\draw (374,110) node [anchor=north west][inner sep=0.75pt]  [font=\small] [align=left] {~\cite{bar04}};

\draw (339,100) node [anchor=north west][inner sep=0.75pt]  [font=\small] [align=left] {~\cite{gavinskyetal}};

\draw (300,70) node [anchor=north west][inner sep=0.75pt]  [font=\small] [align=left] {~\cite{gavinsky2016entangled,girish2022quantum,arunachalam2023one}};
\draw (190,185) node [anchor=north west][inner sep=0.75pt]  [font=\small] [align=left] {~\cite{barequantum20,tfnp}};
\draw (104,182) node [anchor=north west][inner sep=0.75pt]  [font=\small] [align=left] {~\cite{gavinsky2019}};
\draw (105,94) node [anchor=north west][inner sep=0.75pt]  [font=\small] [align=left] {~\cite{tradeoffs23}};
\draw (180.06,115.01) node [anchor=north west][inner sep=0.75pt]  [font=\small] [align=left] {~\cite{Gav08}};
\draw (201.09,67.95) node [anchor=north west][inner sep=0.75pt]  [font=\small,rotate=-33.54] [align=left] {\textbf{\textcolor[rgb]{1,0,0}{This Work}}};

\end{tikzpicture}
    \caption{Quantum versus Classical Communication. Here, an arrow from $A$ to $B$ denotes that $A$ exponentially outperforms $B$ for some task, with solid lines denoting functional tasks and dashed lines denoting relational ones. We use $2$ to denote interactive protocols, $1$ to denote one-way protocols and $\|$ to denote simultaneous protocols.}
    \label{fig:2}
\end{figure}

\vspace{0.2cm}
\noindent \textbf{Acknowledgements:} We thank David Gosset and Robin Kothari for helpful discussions. 
UG and HY are supported by AFOSR award FA9550-23-1-0363, NSF awards CCF-2530159, CCF-2144219, and CCF-2329939, and the Sloan Foundation. NP is supported by the Google PhD Fellowship.
Research at the Perimeter Institute is supported by the Government of Canada through the Department of Innovation, Science and Industry Canada and by the Province of Ontario through the Ministry of Colleges and Universities.

\section{Communication models}

In this section, we define the relevant communication complexity models. We first define a general communication model called the simultaneous message passing model (denoted by parallel bars $\|$) of which the aforementioned $\Dsim,\Qent,\Rent$ models are specific instantiations. See~\Cref{fig:5} for a summary.

\begin{definition}\label{def:SMP}
Let $f : \{0,1\}^n\times \{0,1\}^n\rightarrow \{0,1\}$ be a (partial or total) Boolean function, and $\epsilon \in [0,1]$ be a parameter. A \textbf{simultaneous message passing} protocol $P$ for $f$ involves three parties, Alice, Bob, and a referee. Alice receives $x\in \{0,1\}^n$ as input and Bob receives $y\in \{0,1\}^n$. Alice and Bob send the referee (quantum or classical) message systems $M_A$ and $M_B$ respectively, and the referee subsequently outputs a bit $c=P(x,y)$.

\textbf{Messages.} The messages that Alice and Bob send to the referee can be quantum, denoted by $\Qsim$ or classical, denoted by $\Rsim$.

\textbf{Correctness.} The protocol is $\epsilon$-correct if for all $(x,y)$ in the support of $f$,
\begin{align*}
    \Pr[P(x,y)=f(x,y)] \geq 1-\epsilon \enspace.
\end{align*}
When we drop the subscript $\epsilon$, it means $\epsilon=1/3$. Focusing on the case when Alice and Bob send classical messages and $\epsilon=0$, we obtain the deterministic model of classical simultaneous communication, denoted by $\mathsf{D}\|$.

\textbf{Cost of a protocol.} The cost of the protocol denoted by $\mathrm{cost}(P)$ is defined to be the total number of bits (resp. qubits) sent by Alice and Bob in the $\Rsim$ (resp. $\Qsim$) model. The $\Rsim_\epsilon$ complexity of $f$ is defined as follows
\begin{equation*}
    \Rsim_{\epsilon}(f) = \min_{P: P \text{ is $\epsilon$-correct}}\mathrm{cost}(P)
\end{equation*}
and the $\Qsim_\epsilon$ complexity is analogously defined.

\textbf{Randomness.} Alice and Bob typically have private randomness, but we also consider a variation of the simultaneous message model where we allow public randomness. 
In particular, we allow all three players (Alice, Bob and the referee) to hold a shared random string $r$ of arbitrary length. 
They can then use $r$ as an input to their local operations. 
We label the cost to compute $f$ $\epsilon$-correctly in this model by $\Rsim^{\pub}_\epsilon(f)$ (resp. $\Qsim^{\pub}_\epsilon(f)$) when the messages are classical (resp. quantum).  

\textbf{Entanglement.} We may allow Alice and Bob to share entanglement, denoted by the superscript $*$ and resulting in the models $\Qent$ and $\Rent$ depending on whether the messages to the referee are quantum or classical. 
\end{definition}

\begin{table}
\centering
\begin{tabular}{|l|l|l|l|l|}
\hline
Models                   & Error & Entanglement  & Randomness               & Messages \\
                    &   &   (Alice \& Bob) &              &  to Referee \\ \hline
$\Dsim$                      & 0     & No                          & No                       & Classical           \\ \hline
$\Rsim$                       & 1/3   & No                          & Private             & Classical           \\ \hline
$\Qsim$                      & 1/3   & No                          & Private             & Quantum             \\ \hline
$\Rsim^\pub$  & 1/3   & No                          & Public              & Classical           \\ \hline
$\Qsim^\pub$  & 1/3   & No                          & Public              & Quantum             \\ \hline
$\Rent$  & 1/3   & Yes                         & Subsumed by entanglement & Classical           \\ \hline
$\Qent$   & 1/3   & Yes                         & Subsumed by entanglement & Quantum \\      
\hline
\end{tabular}
\caption{Various models of simultaneous communication}
\label{fig:5}
\end{table}

A variant of the $\Rsim$ model with privacy constraints is the  $\PSM$ (private simultaneous messages) model. Here, the model of communication is identical, but the goal is for the referee to be able to determine $f(x,y)$, but no other information about $x,y$.   
We record a formal definition of $\PSM$ next. 
\begin{definition}\label{def:PSQM}
    A \textbf{private simultaneous message} task is defined by a choice of (partial or total) Boolean function $f:\{0,1\}^n\times \{0,1\}^n\rightarrow \{0,1\}$. Let $\epsilon, \delta \in [0,1]$ be parameters.
    The inputs to the task are $n$-bit strings $x$ and $y$ given to Alice and Bob, respectively.
    Alice then sends a message system $M_0$ to the referee, and Bob sends a message system $M_1$. 
    From the combined message system $M=M_0M_1$, the referee prepares an output bit $z$ whose system is denoted by $Z$.  
    We require the task be completed in a way that satisfies the following two properties.
    \begin{itemize}
        \item \textbf{$\epsilon$-correctness:} There exists a decoding map $\mathbf{V}_{M \rightarrow Z\tilde{M}}$ such that, for all $(x,y)$ in the support of $f$, 
        \begin{align}
            \left \|\tr_{\tilde{M}}(\mathbf{V}_{M \rightarrow Z\tilde{M}} \rho_{M}(x,y) \mathbf{V}_{M \rightarrow Z\tilde{M}}^\dagger ) - \ketbra{f_{x,y}}{f_{x,y}}_Z\right \|_1 \leq \epsilon
        \end{align}
        where $\rho_M(x,y)$ is the density matrix on $M$ produced on inputs $x,y$ and $f_{x,y}=f(x,y)$.
        \item \textbf{$\delta$-security:} There exists a simulator, which is a quantum channel $\mathbfcal{S}_{Z\rightarrow M}(\cdot)$, such that for all $(x,y)$ on which $f$ is defined 
        \begin{align}
            \left \|\rho_{M}(x,y) - \mathbfcal{S}_{Z\rightarrow M}(\ketbra{f_{x,y}}{f_{x,y}}_Z)\right \|_1 \leq \delta.
        \end{align}
        Stated differently, the state of the message systems is $\delta$-close to one that depends only on the function value, for every choice of input.
    \end{itemize}
   
    \textbf{Messages.} When the messages are quantum, we will refer to this model as $\PSQM$ and when the messages are classical, we refer to the model by $\PSM$.

    \textbf{Entanglement.} When Alice and Bob share entanglement, we denote it by the superscript $*$, obtaining the model $\PSQM^*$ when Alice and Bob send quantum messages and $\PSM^*$ when Alice and Bob send classical messages. 

    \textbf{Cost of a protocol.} The cost of the protocol is defined to be the total number of bits sent by Alice and Bob in the $\PSM$ or $\PSM^*$ models. 
    We denote the minimal cost over all $\epsilon=1/3$ correct, $\delta=1/3$ secure protocols by $\PSM(f)$ or $\PSM^*(f)$. 
    The cost measures $\PSQM(f)$ and $\PSQM^*(f)$ are defined similarly, now counting qubits of communication. 
\end{definition}

The setting of non-local quantum computation (NLQC) is similar to the setting of simultaneous message passing; however, there is no referee and instead, Alice and Bob each send a single simultaneous quantum message to each other and then perform local operations. 
Concretely, the setting is shown in \Cref{fig:non-localandlocal}. 
Non-local quantum computation initially appeared as a cheating strategy in quantum position-verification \cite{kent2011quantum,buhrman2014position}, and subsequently has appeared in relation to a number of other subjects \cite{may2020holographic, apel2024security,allerstorfer2024relating, ananth2024unclonable}. 

\begin{figure*}
    \centering
    \begin{subfigure}{0.45\textwidth}
    \centering
    \begin{tikzpicture}[scale=0.6]
    
    \draw[thick] (-1,-1) -- (-1,1) -- (1,1) -- (1,-1) -- (-1,-1);
    
    \draw[thick,mid arrow] (-3.5,-3) to [out=90,in=-90] (-0.5,-1);
    \draw[thick,mid arrow] (3.5,-3) to [out=90,in=-90] (0.5,-1);
    
    \draw[thick,mid arrow] (0.5,1) to [out=90,in=-90] (3.5,3);
    \draw[thick,mid arrow] (-0.5,1) to [out=90,in=-90] (-3.5,3);
    
    \node at (0,0) {$\mathbfcal{N}$};
    
    \node at (0,-5) {$ $};
    
    \end{tikzpicture}
    \caption{}
    \label{fig:local}
    \end{subfigure}
    \hfill
    \begin{subfigure}{0.45\textwidth}
    \centering
    \begin{tikzpicture}[scale=0.4]
    
    \draw[thick] (-5,-5) -- (-5,-3) -- (-3,-3) -- (-3,-5) -- (-5,-5);
    \node at (-4,-4) {$\mathbfcal{V}^L$};
    
    \draw[thick] (5,-5) -- (5,-3) -- (3,-3) -- (3,-5) -- (5,-5);
    \node at (4,-4) {$\mathbfcal{V}^R$};
    
    \draw[thick] (5,5) -- (5,3) -- (3,3) -- (3,5) -- (5,5);
    \node at (4,4) {$\mathbfcal{W}^R$};
    
    \draw[thick] (-5,5) -- (-5,3) -- (-3,3) -- (-3,5) -- (-5,5);
    \node at (-4,4) {$\mathbfcal{W}^L$};
    
    \draw[thick,mid arrow] (-4.5,-3) -- (-4.5,3);
    
    \draw[thick,mid arrow] (4.5,-3) -- (4.5,3);
    
    \draw[thick,mid arrow] (-3.5,-3) to [out=90,in=-90] (3.5,3);
    
    \draw[thick,mid arrow] (3.5,-3) to [out=90,in=-90] (-3.5,3);
    
    \draw[thick] (-3.5,-6) -- (3.5,-6) -- (0,-8) -- (-3.5,-6);
    \draw[thick] (-3.25,-6) -- (-3.25,-5);
    \draw[thick] (3.25,-6) -- (3.25,-5);
    \node at (0,-7) {$\Psi$};
    
    \draw[thick] (-4.5,-6) -- (-4.5,-5);
    \draw[thick] (4.5,-6) -- (4.5,-5);
    
    \draw[thick] (4.5,5) -- (4.5,6);
    \draw[thick] (-4.5,5) -- (-4.5,6);
    
    \end{tikzpicture}
    \caption{}
    \label{fig:non-localcomputation}
    \end{subfigure}
    \caption{(a) Circuit diagram showing the local implementation of a channel $\mathbfcal{N}$. (b) Circuit diagram showing the form of a non-local quantum computation. $\mathbfcal{V}^L$, $\mathbfcal{V}^R$, $\mathbfcal{W}^L$, and $\mathbfcal{W}^R$ are quantum channels. The goal is to simulate the local channel $\mathbfcal{N}$.}
    \label{fig:non-localandlocal}
\end{figure*}
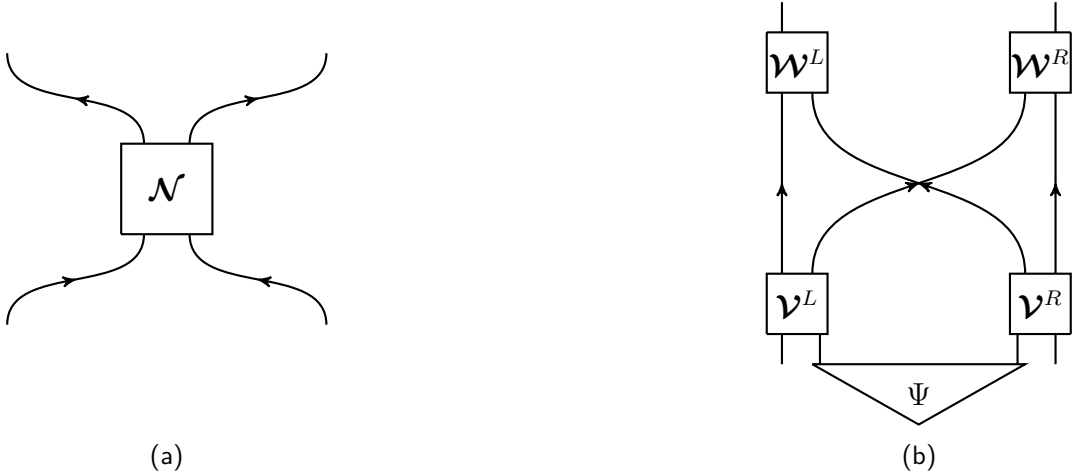

Next, we consider a stronger notion of communication complexity where interactivity is allowed. 
\begin{definition}
Let $f : \{0,1\}^n\times \{0,1\}^n\rightarrow \{0,1\}$ be a function, and $\epsilon \in [0,1]$ be a parameter. 
A \textbf{two-way classical communication} protocol $P$ for $f$ involves two players, Alice and Bob. 
Alice receives $x\in \{0,1\}^n$ as input; Bob receives $y\in\{0,1\}^n$ as input. 
Alice and Bob may additionally share a random string $r$. 
The protocol consists of a sequence of messages passed from Alice to Bob and then Bob to Alice, with Alice eventually outputting a bit $z$. 
The protocol is $\epsilon$ correct if $Pr[f(x,y)=z]\geq 1-\epsilon$ for all $(x,y)$ on which $f$ is defined.
Each message may be computed from the locally held input, the randomness, and any previous messages received by that player.
The cost of a protocol is the number of bits passed between Alice and Bob, maximized over inputs. 
The two way classical communication complexity cost of $f$, $\R_\epsilon(f)$, is defined as the minimal communication cost of any such protocol. 
\end{definition}

\section{Magic lower bounds from communication complexity}

\subsection{Computation and communication models}\label{sec:modeldefinitions}

We start by defining three notions of a Clifford+Magic circuit: the unitary, mixed, and adaptive models.
Throughout this work, we say that a quantum circuit computes a Boolean function $f:\{0,1\}^n\rightarrow \{0,1\}$ with correctness $\epsilon$ if there is a state $\ket{\psi}$ such that running the circuit on $\ket{x}\ket{\psi}$ and measuring the first qubit in the computational basis returns $f(x)$ with probability at least $1-\epsilon$ for all $x$ in the support of $f$. 
The state $\ket{\psi}$ cannot depend on $x$. 

\begin{definition}
    A \textbf{unitary Clifford+Magic} circuit is a quantum circuit composed of Clifford gates along with arbitrary magic gates.
    The cost $\mathcal{M}^{\text{unitary}}_{\epsilon,c_M}(f)$ to compute a Boolean function $f$ in this model is the minimal number of magic gates, each with weight\footnote{The weight of a gate is defined to be the number of qubits on which it acts.} at most $c_M$, appearing in any such circuit that computes $f(x)$ with probability $1-\epsilon$.
    We allow the circuit access to an arbitrary advice state.\footnote{Note that the advice system can both have an arbitrary size, and begin in an arbitrary state.} 
\end{definition}

\begin{definition}
    A \textbf{mixed Clifford+Magic} circuit is a quantum operation $\mathcal{N}$ of the form
    \begin{align}
        \mathcal{N}(\cdot) = \sum_i p_i U_i(\cdot) U_i^\dagger
    \end{align}
    where $\{p_i\}$ is a probability distribution, and $U_i$ is a unitary Clifford+Magic circuit. 
    We consider the magic gate count of a mixed Clifford+Magic circuit to be the worst case magic gate count among the $U_i$. 
    The cost $\mathcal{M}^{\text{mixed}}_{\epsilon,c_M}$ to compute a Boolean function $f$ using a mixed Clifford+Magic circuit is the minimal number of magic gates in any such quantum operation, using gates of weight at most $c_M$, that computes $f$ with probability $1-\epsilon$. 
    We allow access to an arbitrary advice state.
\end{definition}

\begin{definition}
    An \textbf{adaptive Clifford+Magic} circuit is a quantum circuit composed of Clifford gates, arbitrary magic gates, and mid-circuit computational basis measurements.
    Later gate choices may be conditioned on the outcomes of mid-circuit measurements. After a mid-circuit measurement, the choice of the remaining circuit is an arbitrary function of the measurement outcomes so far.
    We consider the cost of an adaptive circuit to be the total number of magic gates plus measurements in the worst-case run of the adaptive circuit. 
    The cost $\mathcal{M}^{\text{adaptive}}_{\epsilon,c_M}$ to compute a Boolean function $f$ using a mixed Clifford+Magic circuit is the minimal cost of any adaptive Clifford+Magic circuit, allowing $c_M$-qubit magic gates and $c_M$-qubit measurements,\footnote{Note that in our convention we count measuring $c_M$ qubits in the computational basis simultaneously in the circuit as a ``single'' measurement. This is somewhat arbitrary, but keeps some constant factors in our eventual lower bound tidy.}, that computes $f$ with probability $1-\epsilon$.
\end{definition}
Note that our adaptive model differs from the one defined in \cite{gosset2025multi}. 
In \cite{gosset2025multi}, the adaptive model allows only ancilla rather than an arbitrary advice state, but then the cost in their model is counted as only the number of magic gates (specifically $T$ gates) rather than the magic gates plus the single qubit measurements.

Comparing these models we see that $\mathcal{M}^{\text{unitary}}(f) \geq \mathcal{M}^{\text{mixed}}(f)$ and $\mathcal{M}^{\text{unitary}}(f) \geq \mathcal{M}^{\text{adaptive}}(f)$, which follows because we could choose to not randomize in the mixed model, or not to measure in the adaptive model.
The adaptive model can simulate the mixed model in the sense that it can prepare $\ket{+}$ states and measure in the computational basis, then control the subsequent circuit off of the measurement outcomes. 
However, the cost $\mathcal{M}^{\text{adaptive}}(f)$ includes the number of such measurements needed as well as the subsequent magic-gate cost while $\mathcal{M}^{\text{mixed}}(f)$ counts only the magic-gate cost, so these quantities may be incomparable. 

\subsection{Lower bounds on magic gate count from communication complexity}\label{sec:magiclowerbounds}

\begin{figure}
\begin{center}
\mbox{\Qcircuit @C=1.5em @R=1em {
\lstick{\ket{x_1}} & \multigate{4}{C^0} & \qw & \qw & \multigate{4}{C^1} & \qw & \dots & & \qw & \qw & \multigate{4}{C^k} & \qw \\
\lstick{\ket{x_1}} & \ghost{C^0} & \qw & \qw & \ghost{C^1} \qw & \qw & \dots & & \qw & \qw &\ghost{C^k} & \qw \\
\lstick{\ket{y_1}} & \ghost{C^0} & \qw & \qw & \ghost{C^1} \qw & \qw & \dots & & \qw & \qw &\ghost{C^k} & \qw \\
\lstick{\ket{y_2}} & \ghost{C^0} & \qw & \qw & \ghost{C^1} \qw & \qw & \dots  & & \qw & \qw  &\ghost{C^k} & \qw \\
\lstick{\ket{\psi}} & \ghost{C^0} & \qw & \gate{M_1} & \ghost{C^1} \qw & \qw & \dots & & \qw & \gate{M_{k}} & \ghost{C^k} & \qw & \meter
}}
\end{center}
\caption{Quantum circuit with $k$ magic gates that computes $f(x,y)$.}\label{fig:layeredcircuit}
\end{figure}
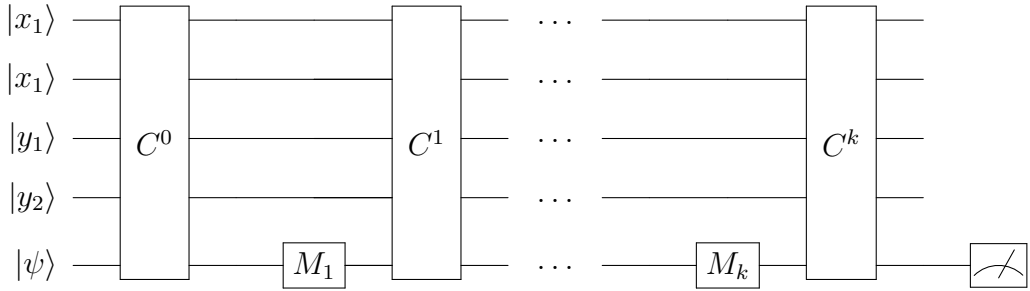

\begin{figure}
\begin{center}
\mbox{
\Qcircuit @C=1.5em @R=1em {
\lstick{\ket{0}} & \multigate{4}{C^0} & \qw & \qw & \multigate{4}{C^1} & \qw & \dots & & \qw & \qw & \multigate{4}{C^k} & \qw \\
\lstick{\ket{0}} & \ghost{C^0} & \qw & \qw & \ghost{C^1} \qw & \qw & \dots & & \qw & \qw &\ghost{C^k} & \qw \\
\lstick{\ket{0}} & \ghost{C^0} & \qw & \qw & \ghost{C^1} \qw & \qw & \dots & & \qw & \qw &\ghost{C^k} & \qw \\
\lstick{\ket{0}} & \ghost{C^0} & \qw & \qw & \ghost{C^1} \qw & \qw & \dots  & & \qw & \qw  &\ghost{C^k} & \qw \\
\lstick{\ket{\psi}} & \ghost{C^0} & \gate{\sigma_1} & \gate{M_1} & \ghost{C^1} \qw & \qw & \dots & & \gate{\sigma_k} & \gate{M_{k}} & \ghost{C^k} & \gate{\sigma_{k+1}} & \meter
}}
\end{center}
\caption{The circuit simulated by the referee in our $\Dsim$ protocol. The construction begins with a circuit (\Cref{fig:layeredcircuit}) that computes $f(x,y)$ from inputs $(x,y)$ along with an advice state. Here, we run the circuit on the all-zeroes input, and Pauli corrections $\sigma_i$ are made just before each magic gate $M_i$ as necessary. One additional Pauli correction is made before the measurement. Each Pauli correction can be computed in the $\Dsim$ model with constant communication, so the total communication cost is a constant times the number of magic gates.
}\label{fig:R||protocol}
\end{figure}
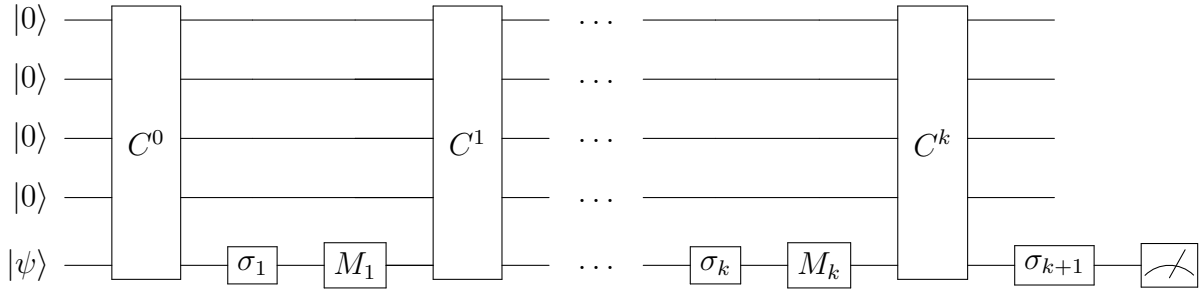

In this section we give our lower bound on the magic gate count from the simultaneous message passing and parity decision tree complexities.
Our first result is the following.
\begin{theorem}\label{thm:unitarymagiclowerbound}
    Let $f$ be a Boolean function. Then the $\Dsim$ and unitary magic gate cost are related as follows. For all $\epsilon < \frac{1}{2}$, 
    \begin{align}
        \frac{1}{4c_M} (\Dsim(f)-2) \leq \mathcal{M}^{\text{unitary}}_{\epsilon}(f)~.
    \end{align}
\end{theorem}

\begin{proof}
The proof was essentially stated in the introduction. 
We repeat the proof here, filling in some details. 

Consider a Clifford+Magic circuit computing $f$ with any probability $p=1-\epsilon>1/2$. 
We decompose the circuit into layers, consisting of either a Clifford circuit or a magic gate.
See \Cref{fig:layeredcircuit}. 
Let the input to the circuit be $\ket{z}\ket{\psi}$, with $z$ the input to $f$ and $\ket{\psi}$ the advice state. 

We consider any division of the input string $z$ into $(x,y)$, and give a $\Dsim$ protocol for computing $f$ with respect to this division. 
The referee runs the circuit on the input $\ket{0}\ket{0}\ket{\psi}$, which we view as $X_A^{\vec{x}}\ket{x}_A X_B^{\vec{y}}\ket{y}_B\ket{\psi}_E$. 
Then after the first Clifford layer is applied, we have
\begin{align}
    C^1_{ABE}X^{\vec{x}}_A\ket{x}_A X^{\vec{y}}_B\ket{y}_B\ket{\psi}_E = \sigma_{ABE}[x,y]C^1\ket{x}_A\ket{y}_B\ket{\psi}_E.
\end{align}
$\sigma_{ABE}[x,y]$ denotes a string of Pauli corrections, which depend on the inputs $(x,y)$. 
At this point we would like to apply the first magic gate. 
Before doing so, we compute the Pauli corrections that act on the same qubits as the magic gate. 
These are determined by at most $2c_M$ parity functions of $(x,y)$: this is because the magic gate acts on at most $c_M$ qubits, and each qubit can have a Pauli correction of the form $X^aZ^b$.\footnote{We can ignore global phases.}
The values of $a$ and $b$ are determined by a parity function of $(x,y)$, with the choice of parity function dependent on the circuit $C^0$. 
For any parity function $p(x,y)=\sum_{i\in S_A}x_i+ \sum_{i\in S_B}y_i$, Alice can send the single bit $\sum_{i\in S_A}x_i$ and Bob the single bit $\sum_{i\in S_B}y_i$, allowing the referee to compute $p(x,y)$, so the communication cost is $2$.

After correcting these Pauli corrections we apply the relevant magic gate, then the next Clifford layer. 
The remaining Pauli corrections conjugate through the second Clifford layer to give further Pauli corrections before the second magic gate. 
We again have Alice and Bob send messages to allow the referee to compute the $2c_M$ parities that determine the needed corrections, then proceed as before. 

This process continues until reaching the end of the circuit. 
Finally, we compute one more parity function to determine if there is a Pauli $X$ correction before the final measurement. 
Correcting this if needed and then measuring, we obtain a sample of the output distribution of the circuit. 
This procedure is illustrated in \Cref{fig:R||protocol}. 

This simulation can be repeated (using the same parity values each time) so that we can determine the output distribution of the final measurement. 
If the outcome is $0$ with probability more than $1/2$ the referee outputs $0$, otherwise we output $1$. 
If the Clifford+Magic circuit is correct with probability $1-\epsilon>1/2$, this yields the correct output. 
The total cost is $4c_M$ times the total number of magic gates, plus 2 for the final measurement, so that $\Dsim(f)\leq 4c_M \cdot \mathcal{M}_{\epsilon<1/2}^{\text{unitary}}(f)+2$. 
Rearranging this gives the claimed lower bound. 
\end{proof}

\begin{remark}\label{remark:postselectunitary}
    We can actually strengthen the computational model lower bounded by $\Dsim$: suppose the circuit model is allowed to post-select onto fixed quantum states. To simulate this in the $\Dsim$ model, we send the Pauli corrections occurring just before the post-selection. This adds 2 bits of communication cost for each qubit of post-selection. Thus $\Dsim(f)$ also lower bounds the number of magic gates plus the number of qubits of post-selection in a Clifford+Magic + post-selection circuit that computes $f$. We can similarly allow for post-selection in the lower bounds below on the mixed Clifford+Magic model, and in the adaptive model. 
\end{remark}

Next, we build on the proof technique used above to bound the mixed Clifford+Magic computational model in terms of a randomized $\Dsim$ communication model. 
\begin{theorem}\label{thm:mixedlowerbound}
    Let $f$ be a Boolean function. Then the $\Rsim^{\pub}$ and mixed unitary magic gate cost are related by 
    \begin{align}
        \frac{1}{4c_M} (\Rsim^{\pub}_{\epsilon}(f)-2) \leq \mathcal{M}^{\text{mixed}}_{\epsilon}(f). 
    \end{align}
\end{theorem}
\begin{proof}
    Consider a mixed Clifford+Magic circuit defined by probabilities $\{p_i\}$ and unitaries $\{U_i\}$. 
    Let the probability with which circuit $U_i$ outputs $f(z)$ given input $z$ be $P_i(z)$. 
    Then the success probability for the mixed circuit is
    \begin{align}
        p_{suc}(z) =\sum_i p_i P_i(z).
    \end{align}
    By assumption, this is larger than $1-\epsilon$. 
    
    The communication protocol is as follows. 
    Alice, Bob and the referee use the public randomness to draw from the set $\{U_i\}$ according to the probabilities $\{p_i\}$. 
    When they draw $U_i$, they run the $\Dsim$ protocol defined in \Cref{thm:unitarymagiclowerbound} so that Alice and Bob send the parities needed for the referee to simulate the Clifford+Magic circuit $U_i$. 
    Now however, the referee just samples from this circuit once and returns the output. 
    This will be correct with probability $P_i$, so the overall success probability of the communication protocol is just $\sum_i p_iP_i$ as before. 
    This is larger than $1-\epsilon$ as needed. 

    Finally, note that the bits sent in this randomized protocol is the number of bits sent in the protocol for the selected $U_i$, which is $4c_M$ times the number of magic gates in $U_i$, plus 2. 
    The worst case communication cost is then set by the number of magic gates maximized over the $U_i$, which corresponds to our definition of $\mathcal{M}^{\text{mixed}}_\epsilon(f)$. 
    Thus $\Rsim^{\pub}_{\epsilon}(f) \leq 4c_M \mathcal{M}^{\text{mixed}}_{\epsilon,c_M}(f) +2$, which gives the claimed lower bound.
\end{proof}

Finally we consider the adaptive Clifford+Magic model. 
Recall that we defined the adaptive model to allow mid-circuit measurements and an arbitrary advice state, and the cost to be the number of magic gates plus the number of single qubit measurements. 
We will show this cost is lower bounded by the two-way communication complexity. 

\begin{theorem}\label{thm:adaptivelowerbound}
    Let $f$ be a Boolean function. 
    Then the two-way communication complexity $\R_\epsilon(f)$ and the cost $\mathcal{M}_{\epsilon,c_M}^{\text{adaptive}}(f)$ are related by
    \begin{align}
        \frac{1}{2c_M}\left(\R_\epsilon(f)-1\right) \leq \mathcal{M}_{\epsilon,c_M}^{\text{adaptive}}(f)
    \end{align}
\end{theorem}
\begin{proof}
    We use a similar strategy as in the last two theorems to build a communication protocol from the adaptive circuit. 
    Consider a decomposition of the adaptive circuit into layers. 
    Each layer may include arbitrary Clifford gates but only one magic gate or measurement, which occurs as the first gate in the layer. 
    Alice prepares the advice state $\ket{\psi}_E$ and runs the circuit on the input $\ket{x}_A\ket{0}_B$ input, which we view as $\ket{x}_A X_B^{\vec{y}}\ket{y}_B$. 
    Alice runs the first Clifford layer, giving
    \begin{align}
        C^1_{ABE}X_B^{\vec{y}}\ket{x}\ket{y}\ket{\psi}_E = \sigma_{ABE}[y]C^1_{ABE}\ket{x}_A\ket{y}_B\ket{\psi}_E
    \end{align}
    Suppose the first non-Clifford operation is a magic gate. 
    Then, Bob computes the identities of the Pauli corrections acting on the wires that magic gates acts on and sends this to Alice. 
    This costs $2c_M$ bits of communication, since recall each magic gates acts on at most $c_M$ qubits, and for each qubit we must communicate whether there is an $X$ correction and a $Z$ correction. 
    Alternatively, suppose the first non-Clifford operation is a measurement, which again may act on $c_M$ qubits.
    According to our model we assume the measurement is in the computational basis. 
    Then Bob sends the Pauli $X$ corrections acting on the measured wires, Alice performs the appropriate corrections before making the measurement, and Alice then sends back to Bob the $c_M$ bits of measurement outcome. 
    The total communication cost of the measurement is $2c_M$, as with the magic gate. 
    Alice and Bob then both determine the next layer of the circuit based on the measurement outcomes. 

    This procedure repeats for every layer of the circuit, giving a total cost of $2c_M$ multiplied by the number of magic gates or mid-circuit measurements. 
    The final measurement that determines $f(x,y)$ requires an additional Pauli $X$ correction, contributing $+1$ to the communication cost.
    The measurement outcome determines $f(x,y)$ with probability $1-\epsilon$.
    Overall then we have that
    \begin{align}
        \R_\epsilon(f) \leq 2c_M \mathcal{M}^{\text{adaptive}}_{\epsilon,c_M}(f) + 1
    \end{align}
    which gives the claimed lower bound on $\mathcal{M}^{\text{adaptive}}_{\epsilon,c_M}(f)$.
\end{proof}

\subsection*{Bounds from parity decision trees}

Finally, our lower bounds on the unitary and mixed models can be strengthened to be in terms of a classical computational model known as a parity decision tree (PDT). 
Lower bounds from parity decision tree's were proven independently in \cite{gosset2025multi}. 
We point out here that the lower bounds from PDT complexity, which can be stronger than the bounds in terms of communication complexity, can also be recovered using our proof technique.\footnote{Note that we were led to consider if our technique gave lower bounds from PDT's after discussing these results with the authors of \cite{gosset2025multi}.} 

We define deterministic and randomized variants of parity decision tree's, before discussing the lower bounds. 
\begin{definition}
    Consider a Boolean function $f:\{0,1\}^n\rightarrow \{0,1\}$. 
    A \textbf{non-adaptive parity decision tree} ($\PDT^{\na}$) of depth $k$ computing $f$ is a function $g:\{0,1\}^k\rightarrow \{0,1\}$ such that $f(x)=g(p_1,...,p_k)$, where each $p_i$ is a parity function. 
    The non-adaptive parity decision tree complexity of $f$ is the minimal $k$ such that there is PDT of depth $k$ that computes $f$. 
\end{definition}

\begin{definition}
    Consider a Boolean function $f:\{0,1\}^n\rightarrow \{0,1\}$.
    A \textbf{randomized non-adaptive parity decision tree} ($\RPDT^{\na}$) of depth $k$ computing $f$ is a probability distribution over a set of non-adaptive parity decision tree's all of depth at most $k$. 
    We say the RPDT computes $f$ with probability $1-\epsilon$ if for every choice of input $x$, the RPDT outputs $f(x)$ with probability at least $1-\epsilon$. 
\end{definition}

Regarding the unitary model, from the proof of \Cref{thm:unitarymagiclowerbound}, we can observe that 
\begin{align}
    \frac{1}{2c_M}(\PDT^{\na}(f) - 1) \leq \mathcal{M}^{\text{unitary}}_{\epsilon<1/2}(f).
\end{align}
This follows because we can observe in the proof of~\Cref{thm:unitarymagiclowerbound} that the communication from Alice and Bob purely consists of parities of $x,y$ (recall that these were needed to do the required Pauli corrections), furthermore, these parity functions depend only on the circuit and not on the input.

Specializing to the case of $T$ gates, we have $c_M=1$ and can observe that since $T$ commutes with $Z$, we can actually leave the $Z$ Paulis uncorrected. 
This halves the number of parities, and gives the bound
\begin{align}
    \PDT^{\na}(f) - 1 \leq \mathcal{T}^{\text{unitary}}_{\epsilon<1/2}(f).
\end{align}
This is exactly one of the lower bounds proven by another technique in \cite{gosset2025multi}.

Regarding the mixed model, from the proof of \Cref{thm:mixedlowerbound} we can deduce that
\begin{align}
    \frac{1}{2c_M}(\RPDT^{\na}_\epsilon(f)-1) \leq \mathcal{M}^{\text{mixed}}_\epsilon(f).
\end{align}
To understand why, we first view the $\Dsim$ protocol associated with a single $U_i$ as defining a \emph{randomized} parity decision tree (note that in \Cref{thm:unitarymagiclowerbound} we constructed a deterministic PDT). 
Specifically, we consider the same set of parities as before, which allow the circuit to be simulated by correcting Paulis as needed. 
Now however, we have the decision tree sample an output from the distribution defined by the final measurement. 
Thus each $U_i$ is associated to a $\RPDT^{\na}$ which outputs $f(x,y)$ with the same probability as running $U_i$ and measuring the output qubit. 
Now we add an additional randomization step, where we sample $U_i$ with probability $p_i$.
This defines a new $\RPDT^{\na}$, which now outputs $f(x,y)$ with the same probability as the mixed circuit $\sum_i p_i U_i (\cdot) U_i^\dagger$.
Since the number of parities needed to simulate the worst case $U_i$ is $2c_M\cdot \mathcal{M}^{\text{mixed}}_\epsilon(f)+1$, we obtain the above bound. 
Again we can specialize this to the $T$ gate case and obtain
\begin{align}
    \RPDT^{\na}_\epsilon(f)-1\leq \mathcal{T}^{\text{mixed}}_\epsilon(f).
\end{align}
This matches a result in \cite{gosset2025multi}. 

\subsection{Lower bounds for concrete unitary operators}

To illustrate our lower bound technique, we prove tight magic-count bounds on implementing some unitary operations. The approach is the following: to prove a magic-count lower bound on implementing a unitary $U$, we show that the unitary $U$ can be used to efficiently compute a Boolean function $f$ with little-to-no magic overhead. Then, using known lower bounds on the communication complexity of $f$ along with our lower bounds (in particular \Cref{thm:unitarymagiclowerbound,thm:mixedlowerbound}), we obtain magic-count lower bounds for $U$. 

Although the magic-count measures $\mathcal{M}^{\text{unitary}}$ and $\mathcal{M}^{\text{mixed}}$ were defined only for Boolean functions, they have natural extensions to general unitary operators. 
$\mathcal{M}^{\text{unitary}}$ and $\mathcal{M}^{\text{mixed}}$ are defined with respect to the same set of allowed operations as before, but now the requirement is that a target unitary be implemented to within $\epsilon$ distance in diamond norm. 
Notice that if a circuit using unitary $U$ computes $f$ with probability 1, a circuit with $U$ replaced with $U'$ satisfying $\Vert U-U'\Vert_\diamond \leq \epsilon$ will compute $f$ with probability at least $1-\epsilon$. 
This means in particular that if $U$ computes $f$ (with no magic overhead) exactly, then $\mathcal{M}^{\text{unitary}}_{\epsilon}(U) \geq \mathcal{M}^{\text{unitary}}_\epsilon(f)$ and $\mathcal{M}^{\text{mixed}}_{\epsilon}(U) \geq \mathcal{M}^{\text{mixed}}_\epsilon(f)$. 


\paragraph{Generalized Toffoli gates.} An $n$-qubit generalized Toffoli gate computes the following:
\[
    \mathrm{Toffoli}_n : \ket{x_1,\ldots,x_n,b} \mapsto \ket{x_1,\ldots,x_n,b \oplus \bigwedge_{i=1}^n x_i}~.
\]
In other words, it XORs the AND of the first $n$ bits into the target qubit. 

We construct a quantum circuit that uses the generalized Toffoli gate and computes the equality function. 
Consider the circuit $C$ that acts on $2n+1$ qubits (labeled $\mathsf{A}_1,\ldots,\mathsf{A}_n,\mathsf{B}_1,\ldots,\mathsf{B}_n,\mathsf{C}$, and assuming the input is of the form $\ket{x,y,0}$ where $x,y \in \{0,1\}^n$, computes:
\begin{enumerate}
    \item For each $i$, apply CNOT with control on qubit $\mathsf{A}_i$ and target on qubit $\mathsf{B}_i$ to obtain $\ket{x_i,x_i \oplus y_i}$. Apply X on $\mathsf{B}_i$ to obtain $\ket{x_i \oplus y_i \oplus 1}$. 
    \item Apply $\mathrm{Toffoli}_n$ controlled on qubits $\mathsf{B}_1,\ldots,\mathsf{B}_n$ and with target qubit $\mathsf{C}$. 
\end{enumerate}
The circuit computes the equality function. Thus we can obtain magic-count lower bounds for $\mathrm{Toffoli}_n$ via communication complexity lower bounds for equality.

\begin{lemma}
    We have that
        \[\mathcal{M}^{\text{mixed}}_{\epsilon}(\mathrm{Toffoli}_n) = \Omega( \min \{ \log 1/\epsilon, n \} )
    \]
    and for all $\epsilon < \frac{1}{2}$,
    \[
    \mathcal{M}^{\text{unitary}}_{\epsilon}(\mathrm{Toffoli}_n) = \Omega(n)~.
    \]
\end{lemma}
\begin{proof}
First we prove the mixed Clifford+Magic lower bound. Since the generalized Toffoli gate can be used to compute the equality function with no magic overhead, we have $\mathcal{M}^{\text{mixed}}_{\epsilon}(\mathrm{Toffoli}_n) \geq \mathcal{M}^{\text{mixed}}_{\epsilon}(\mathrm{Equal}_n)$. By \Cref{thm:mixedlowerbound}, 
\[
    \mathcal{M}^{\text{mixed}}_{\epsilon}(\mathrm{Toffoli}_n) \geq \mathcal{M}^{\text{mixed}}_{\epsilon}(\mathrm{Equal}_n) \geq \frac{1}{4c_M} \Big (\Rsim^{\pub}_{\epsilon}(\mathrm{Equal}_n)-1 \Big)~.
\]
Now, the $\Rsim^{\pub}_\epsilon$ complexity of $\mathrm{Equal}_n$ is well known to be $\Omega( \min \{ \log 1/\epsilon, n \} )$~\cite{dale}. 

Next we prove the unitary Clifford+Magic lower bound. By \Cref{thm:unitarymagiclowerbound}, we have that
\[
    \mathcal{M}^{\text{unitary}}_{\epsilon}(\mathrm{Toffoli}_n) \geq \mathcal{M}^{\text{unitary}}_{\epsilon}(\mathrm{Equal}_n) \geq \frac{1}{4c_M} \Big( \Dsim (\mathrm{Equal}_n) - 2 \Big)~.
\]
It is known that $\Dsim (\mathrm{Equal}_n) = \Omega(n)$~\cite{KN}, which concludes the proof.
\end{proof}

Note that the Toffoli gate is also studied in \cite{gosset2025multi}, where they prove the same lower bound along with a matching upper bound. 

\paragraph{Quantum multiplexer.} The $n$-qubit quantum multiplexer computes the following. Let $x \in \{0,1\}^n$, let $i$ be a $\lceil \log_2 n \rceil$-bit index, and let $b \in \{0,1\}$. Then
\[
    \mathrm{Multiplex}_n: \ket{i,x,b}  \mapsto \ket{i,x_1,\ldots,x_{i-1},b,x_{i+1},\ldots,x_n,x_i}~.
\]
In other words, controlled on the index register $\ket{i}$, the multiplexer swaps the $i$'th bit of $\ket{x}$ with the target register $\ket{b}$. 

Clearly, the multiplexer can be used to compute the index function $\mathrm{Index}_n(i,x) = x_i$ where $x \in \{0,1\}^n$ and $i \in \{0,1\}^{\lceil \log_2 n \rceil}$. We obtain linear lower bounds on the magic-count on mixed Clifford+Magic implementations of the multiplexer.

\begin{lemma}
    We have that
    \[
        \mathcal{M}^{\text{mixed}}_{\epsilon}(\mathrm{Multiplex}_n) \geq \frac{1}{4c_M} \Big( (1 - h(\epsilon)) n - 1 \Big)
    \]
    where $h(\epsilon) = -\epsilon \log \epsilon - (1 - \epsilon)\log(1-\epsilon)$ is the binary entropy function.  Furthermore, we have the upper bound
    \[
    \mathcal{M}^{\text{unitary}}_{\epsilon = 0}(\mathrm{Multiplex}_n) = O(n)~.
    \]
\end{lemma}
\begin{proof}
We first prove the lower bound.
Since the quantum multiplexer can be used to compute the index function, we have $\mathcal{M}^{\text{mixed}}_{\epsilon}(\mathrm{Multiplex}_n) \geq \mathcal{M}^{\text{mixed}}_{\epsilon}(\mathrm{Index}_n)$. By \Cref{thm:mixedlowerbound}, 
\[
    \mathcal{M}^{\text{mixed}}_{\epsilon}(\mathrm{Multiplex}_n) \geq \mathcal{M}^{\text{mixed}}_{\epsilon}(\mathrm{Index}_n) \geq \frac{1}{4c_M} \Big (\Rsim^{\pub}_{\epsilon}(\mathrm{Index}_n)-1 \Big)~.
\]
Now, the $\Rsim^{\pub}_\epsilon$ complexity of the $n$-bit index function is at least the randomized one-way communication complexity (from Bob to Alice) of the index function: any $\Rsim^{\pub}_\epsilon$ protocol for the index function (where Alice receives $i$, Bob receives $x$, and Alice, Bob, referee receive a uniformly random string $r$) can be converted into a one-way protocol where Bob simply sends Alice the message he would've sent to the referee. It is known that the one-way randomized communication complexity of the index function is at least $(1 - h(\epsilon))n$, even with quantum communication~\cite{nayak1999optimal}, since Bob's message to Alice would define a random access code for his input. 

We now prove the upper bound by induction. We actually construct a \emph{controlled} quantum multiplexer
\[
    \mathrm{cMultiplex}_n = \ketbra{0}{0} \otimes I + \ketbra{1}{1} \otimes \mathrm{Multiplex}_n~.
\]
Note that $\mathrm{cMultiplex}_n$ operates on $2+\lceil \log_2 n \rceil + n$ qubits. Clearly, a controlled multiplexer can be used to implement a non-controlled multiplexer (by setting the control qubit to $\ket{1}$).

Suppose that the $2^k$-qubit controlled multiplexer can be implemented with magic-count $g(k)$ where the magic gates have maximum width $3$. Then the $2^{k+1}$-qubit controlled multiplexer can be implemented as follows. Let the control qubit be denoted $\mathsf{C}$, the $(k+1)$-bit index register be denoted $\mathsf{I}_1,\ldots,\mathsf{I}_{k+1}$, the $2^{k+1}$-bit array register be denoted $\mathsf{X}_1,\ldots,\mathsf{X}_{2^{k+1}}$, and the target qubit be denoted $\mathsf{T}$. We describe the circuit.

\begin{enumerate}
    \item Apply X to $\mathsf{I}_1$. Apply a Toffoli controlled on $\mathsf{C}$ and $\mathsf{I}_1$ with target $\mathsf{A}_1$, an ancilla qubit. Apply X to $\mathsf{I}_1$.
    \item Apply a Toffoli controlled on $\mathsf{C}$ and $\mathsf{I}_1$ with target $\mathsf{A}_2$, an ancilla qubit.
    \item Perform the $2^k$-qubit controlled multiplexer with control qubit $\mathsf{A}_1$, index register $\mathsf{I}_2,\ldots,\mathsf{I}_{k+1}$, the first half of the array $\mathsf{X}_1,\ldots,\mathsf{X}_{2^k}$, and the target qubit $\mathsf{T}$.
    \item Perform the $2^k$-qubit controlled multiplexer with control qubit $\mathsf{A}_2$, index register $\mathsf{I}_2,\ldots,\mathsf{I}_{k+1}$, the second half of the array $\mathsf{X}_{2^k+1},\ldots,\mathsf{X}_{2^{k+1}}$, and the target qubit $\mathsf{T}$.
    \item Uncompute the $\mathsf{A}_1,\mathsf{A}_2$ registers.
\end{enumerate}
Intuitively, the ancilla qubits $\mathsf{A}_1$ and $\mathsf{A}_2$ store whether the $2^k$-size controlled multiplexer should be implemented on the left or right half of the array. At most one of these ``half'' multiplexers will be activated. This construction cleanly implements the $2^{k+1}$-qubit controlled multiplexer. The magic-count satisfies
\[
    g(k+1) \leq 2g(k) + 4~.
\]
We also have $g(1) = O(1)$. Thus the magic count of the $2^{k+1}$-size controlled multiplexer is $g(k+1) \leq O(2^{k+1})$, as desired.
\end{proof}



\section{Communication upper bounds from magic depth}

\subsection{The private simultaneous message passing model and NLQC}\label{sec:PSMandNLQCdefinitions}

A key set of tools we make use of to prove our upper bound are techniques from \cite{speelman2015instantaneous}, which were originally used to prove upper bounds on NLQC. 
Specifically, \cite{speelman2015instantaneous} proves the following theorem. 
\begin{theorem}\label{thm:TdepthNLQC}
    Consider a unitary $U_{AB}$ which can be expressed as a Clifford$+T$ circuit, with $T$-depth at most $d$, and which acts on $n$ qubits. Then $U_{AB}$ can be implemented as an NLQC using communication of at most $O((68n)^d)$ bits and at most $O((68n)^d)$ shared EPR pairs. 
\end{theorem}
It will be useful later to introduce the key techniques used in \cite{speelman2015instantaneous} to prove this theorem. 

One idea used in the proof of this theorem is the garden-hose model \cite{buhrman2013garden}.
The garden-hose model is most easily described in terms of the following setting. 
Alice and Bob are neighbours, and share a fence. 
Alice has an input string $x\in\{0,1\}^n$, while Bob has an input string $y\in\{0,1\}^n$. 
Alice has a tap, which she can turn on to produce a flow of water. 
Alice and Bob share a number of pipes which connect their yards, and they have hoses that they can use to connect pipes to one another, or to connect the tap to a pipe. 
Alice and Bob wish to compute a Boolean function $f(x,y)$, with the outcome determined by where the water spills. 
Typically, the model is defined so that water spilling on Alice's side indicates $f(x,y)=0$, while water spilling on Bob's side indicates $f(x,y)=1$. 
The garden-hose model can also be formalized in terms of path connectivity in particular form of graph, see \cite{buhrman2013garden}, though we won't introduce this formalization here.
The minimal number of pipes needed in a garden-hose protocol that computes $f(x,y)$ is the garden-hose complexity of $f$, which we denote by $GH(f)$. 

In the quantum context the garden-hose model appears as a description of concatenated teleportations in some settings. 
In particular, consider an unknown quantum state $\ket{\psi}$, which plays the role of the tap in the garden-hose description. 
Alice and Bob share a set of EPR pairs between them, which play the role of the pipes. 
Alice and Bob can then make Bell basis measurements, which act on either two ends of EPR pairs they hold in their own labs, or (in Alice's case) on the input state plus the end of one EPR pair.
To see why the water-flow analogy of the garden-hose model is relevant, consider that after the input state is measured with one EPR pair, the state has moved to the other end of the EPR pair, up to Pauli corrections. 
Each subsequent measurement moves the state to the other end of the measured EPR pair. 
To an observer with access to the measurement outcomes, it is as if the state is flowing along the path determined by the pipes in the garden-hose picture. 

The following lemma related to the garden-hose model is needed in the proof of theorem \ref{thm:TdepthNLQC}.
The lemma is proven in \cite{speelman2015instantaneous}.
\begin{lemma}\label{lemma:GHgrowth}
    Let $f$ be a Boolean function with garden-hose complexity $GH(f)$. Suppose Alice initially has the state $P^{f(x,y)}\ket{\psi}$ where $x$ is known to Alice and $y$ is known to Bob.
    Then the following two statements hold:
    \begin{enumerate}
        \item There exists an instantaneous protocol (no communication) which uses $2GH(f)$ EPR pairs after which Alice holds $X^{g(\hat{x})}Y^{h(\hat{x})}\ket{\psi}$, where $\hat{x}$ consists of $x$ and $2GH(f)$ bits that describe Alice and Bob's measurement outcomes. 
        \item The garden hose complexities of $g$ and $h$ are at most linear in the complexity of $f$, 
        \begin{align}
            GH(g) &\leq 4GH(f)+1, \nonumber \\
            GH(h) &\leq 11GH(f)+2.
        \end{align}
    \end{enumerate}
\end{lemma}

We also need the following lemma from~\cite{speelman2015instantaneous}.

\begin{lemma}\label{lemma:xor}
	Let $f_1,\ldots,f_m$ be Boolean functions and $c\in\{0,1\}$ be any bit. Then, for $f=f_1\oplus\ldots\oplus f_m\oplus c$, we have $GH(f)\le 4\sum_{i=1}^m GH(f_i)+1$.  
\end{lemma}

\subsection{Transforming \texorpdfstring{$\Qent$}{TEXT} protocols into \texorpdfstring{$\PSM^*$}{TEXT} protocols}\label{sec:Q||toPSM}
Consider an arbitrary $\Qent$ protocol. We can view the referee's actions as first applying a unitary $U$ and then measuring the first qubit to determine $f(x,y)$. In this section, we show a technique to convert such protocols into $\PSM$ protocols. When $U$ has low $T$-depth, this transformation will be efficient.

\begin{theorem}\label{thm:Q||toPSM}
    Consider an $\epsilon$-correct $\Qent$ protocol for function $f$, which uses $m$ qubits of message. 
    Suppose that this protocol involves the referee applying a $T$-depth-$d$ unitary to the messages received from Alice and Bob, along with at most $a$ qubits of ancilla and then measuring the first qubit to return the output. 
    Then there is a $\PSM_{\epsilon, \delta=2\epsilon}^*$ protocol for $f$ which uses $O((68(m+a))^d)$ qubits of communication and entanglement.
\end{theorem}

\begin{proof}
To begin, suppose Alice and Bob have already executed their own actions in the $\Qent$ protocol, and now hold message system $M_A$ and $M_B$. 
Instead of sending those message systems to the referee, Alice keeps $M_A$ and Bob teleports $M_B$ to Alice without revealing the Pauli corrections. 
Alice will now attempt to execute the unitary $U_{M_AM_BE}$ that would otherwise be executed by the referee, where $E$ is some advice system introduced by the referee. 
The issue with this is that Alice only has Bob's state up to Pauli corrections. 
To deal with these Pauli corrections, we will use the Clifford$+T$ decomposition of $U$ and track how the Pauli corrections evolve through the layers of the circuit.
    
To start with, Alice executes the first Clifford$+T$ layer. 
The initial teleportation done by Bob leads to Pauli corrections on the inputs, which conjugate to a potentially different Pauli corrections after the first Clifford circuit -- but these are known to Bob since he knows the circuit. We now see how these Pauli corrections pass through the layer of $T$ gates. From the relations
\begin{align}
        TX=PXT,\qquad TZ = ZT
\end{align}
we see that the Pauli corrections commute through while potentially incurring $P$ gate corrections\footnote{$P$ here denotes the phase gate, $P=\begin{pmatrix}1 &0 \\ 0 & i\end{pmatrix}$.}. We will exchange these $P$ corrections for Pauli corrections using \Cref{lemma:GHgrowth}, at the expense of creating somewhat more complex Pauli corrections.
    
In more detail, the first part of \Cref{lemma:GHgrowth} shows that we can use a garden-hose gadget to undo the conditional $P$ gates. 
Initially, the function that determines whether there is a $P$ correction to be done has garden-hose complexity 1, as it is completely known to Bob. 
By~\Cref{lemma:GHgrowth}, the resulting Pauli corrections after applying the first Clifford$+T$ layer are also constant. 
As a result, Alice has implemented the first Clifford$+T$ layer, up to Pauli corrections of constant garden-hose complexity.

Next, Alice needs to apply the second Clifford$+T$ layer. 
She begins by first applying the needed Clifford circuit. 
The Pauli corrections from the previous round commute through the Clifford and transform into new Pauli corrections. 
Whether there is a particular Pauli correction or not after the Clifford depends on the XOR of a subset of the corrections appearing before the Clifford, which we argued had constant garden-hose complexity. 
By~\Cref{lemma:xor}, these new corrections have garden-hose complexity at most $4(m+a)$, since there are at most $m+a$ qubits. 
Now Alice applies the layer of $T$ gates. 
These again potentially lead to $P$ corrections on the wires after the $T$ gates. 
These are corrected similarly to before, using garden-hose gadgets, whose complexity are now $O(m+a)$. 

One can continue in this way, applying Clifford$+T$ layers and handling $P$ corrections using increasingly expensive garden-hose gadgets. 
The accounting for the total entanglement cost of these gadgets matches the cost when implementing the full unitary $U$, and so is as given in \Cref{thm:TdepthNLQC}: the total cost is $O((68(m+a))^d)$, where $m$ is the number of qubits of message and $a$ is the number of qubits of advice used by the referee.
In our case, after the final Clifford$+T$ layer, we apply an additional Clifford layer and then measure a single qubit. 

Alice now sends all of her measurement outcomes, from both measuring the final qubit and her Bell basis measurements made in the execution of the garden-hose gadgets, to the referee. 
Bob sends all of his measurement outcomes, all of which come from Bell basis measurements. 
Alice throws away all of her unmeasured qubits. From the Bell basis measurements, the referee can determine if there was a Pauli $X$ correction on the final measured qubit or not, and hence learns the corrected measurement outcome. 
This is $\epsilon$ correct if the $\Qent$ protocol was $\epsilon$ correct. 
    
At this point, we have already shown that a $\Qent$ protocol with a referee that acts in constant $T$ depth may be efficiently transformed into an $\Rent$ protocol.
To show this is in fact a $\PSM^*$ protocol, we need to show $\delta$ security. 
For this, we consider that all of the bits sent to the referee were from Bell basis measurements, call them $\vec{r}=(r_1,...,r_k)$, except one bit $s$, which came from Alice measuring a single qubit of output of $U_{M_AM_BE}$. We first observe that the Bell basis measurement outcomes $\vec{r}$ are distributed as a uniformly random bit-string in $\{0,1\}^{|r|}$. 
To design a simulator, consider that the message is of the form
\begin{align}
    \rho_M(x,y)&=\frac{1}{2^{|r|}}\sum_r X^{p(r)} \sigma(x,y)X^{p(r)}\otimes \ketbra{r}{r}, \nonumber \\
    \sigma(x,y) &= \alpha(x,y) \ketbra{f}{f} + (1-\alpha(x,y)) \ketbra{f\oplus 1}{f\oplus 1}.
\end{align}
Here $p(r)$ is a parity function which determines if there is a Pauli $X$ correction on the measured qubit. 
The probabilities $\alpha(x,y)$ can in general leak information about $(x,y)$, but we have that $\alpha(x,y) \geq 1-\epsilon$ for all $(x,y)$ which will ensure this leaked information is small. 
In particular we define the simulator distribution to be
\begin{align}
    \text{Sim}(f) = \frac{1}{2^{|r|}}\sum_r X^{p(r)} \ketbra{f}{f}  X^{p(r)}\otimes \ketbra{r}{r}.
\end{align}
Then to check security, we just need to calculate the trace distance between the message distribution and the simulator distribution, 
\begin{align}
    \Vert \rho_M(x,y) - \text{Sim}_M(f) \Vert_1 &= \vabs{\frac{1}{2^{|r|}}\sum_r X^{p(r)}(\sigma(x,y)-\ketbra{f}{f})X^{p(r)}\otimes \ketbra{r}{r} }_1 \nonumber \\
    &= \frac{1}{2^{|r|}}\sum_r \Vert \sigma(x,y)-\ketbra{f}{f} \Vert_1 \nonumber \\
    &= \frac{1}{2^{|r|}} \sum_r \|(\alpha(x,y)-1)\ketbra{f}{f} + (1-\alpha)\ketbra{f\oplus 1}{f\oplus 1} \|_1\nonumber \\
    &\leq \frac{1}{2^{|r|}} \sum_r 2 |1-\alpha(x,y)| \nonumber \\
    &\leq 2\epsilon
\end{align}
so that the protocol is $\delta=2\epsilon$ secure, as claimed. 
\end{proof}

\begin{remark} We remark that our simulation of $\Qent$ by $\Rent$ works for relational problems as well, without the privacy condition -- the only difference is that Alice will send the referee the measurement outcomes of a subset of qubits as opposed to a single qubit.
\end{remark}

\begin{remark}
 \Cref{thm:Q||toPSM} also gives a $T$-depth lower bound of 
    \begin{align}
        T\text{-depth}(\Lambda_f) = \Omega\left(\frac{\PSM^*(f)}{\Qent(f)}\right)
    \end{align}
    where $T\text{-depth}(\Lambda_f)$ denotes the $T$-depth of any circuit which implements the measurement $\Lambda$ applied by the referee in the $\Qent$ protocol. 
    This relates the second of Gavinsky's problems mentioned in the introduction to the problem of proving $T$-depth lower bounds: a separation between $\Rent$ and $\Qent$ (and hence between $\PSM^*$ and $\Qent$) would prove a $T$-depth lower bound on referee's measurement in the $\Qent$ protocol. 
\end{remark}

\subsection{Separating \texorpdfstring{$\Rent$}{TEXT} and \texorpdfstring{$\Rtwo$}{TEXT}}

In this section, we describe the problems used by~\cite{girish2022quantum,arunachalam2023one} to separate $\Qent$ and $\Rtwo$ and describe how the referee's actions in the $\Qent$ protocols can be implemented with constant $T$-depth. This along with our results immediately implies a similar separation between $\Rent$ and $\Rtwo$.

\paragraph*{Forrelation.} 

We will now describe the Forrelation-based approach used by~\cite{girish2022quantum} to separate $\Qent$ and $\Rtwo$. Let $n$ be a power of 2. Define the forrelation of a string $x\in\{-1,1\}^{n}$ as
\begin{align*}
    \forr(x) := \frac{1}{n}\bra{x_1}H^{\otimes n}\ket{x_2}
\end{align*}
where $x_1$ is the first half of $x$ and $x_2$ is the second half of $x$. Define a communication problem as follows.
\begin{definition}
Alice gets $x\in\{-1,1\}^n$ and Bob gets $y\in \{-1,1\}^n$, where $n$ is a power of 2. The goal of the players is to output $f(x,y)$ defined by
\begin{align*}\label{def:forrelation}
    f(x,y) = \begin{cases}
        -1 \quad \text{if}\,\, \forr(x\cdot y)\geq \alpha \\
        +1 \quad \text{if}\,\, \forr(x\cdot y)\leq \alpha/2
    \end{cases}
\end{align*}
where $\alpha>0$ is a constant. Here, $x\cdot y$ denotes the point-wise product of $x$ and $y$.
\end{definition}

A variant of this problem with $\alpha=\Theta(1/\log N)$ was originally studied by~\cite{girish2022quantum} who used it to separate $\Qent$ and $\Rtwo$. For a small constant $\alpha$, this problem was studied by~\cite{girish2025} who showed an $\Rtwo$ lower bound of $\tilde{\Omega}(n^{1/4})$, as well as a $\Qent$ upper bound of  $O(\log n)$ where the referee's actions can be implemented in $T$-depth 2. This gives us the desired result. 

\paragraph*{ABCD Problem.}
\begin{definition}
    Alice gets $A,C\in \mathrm{SU}(n)$ and Bob gets $B,D\in \mathrm{SU}(n)$ and their goal is to output $f(x,y)$ defined by
\begin{align*}
    f(x,y) = \begin{cases}
        -1 \quad \text{if}\,\, \Tr(ABCD)\geq 0.9n \\
        +1 \quad \text{if}\,\, \Tr(ABCD)\leq 0.1n
    \end{cases}
\end{align*}
\end{definition}
It was shown by~\cite{arunachalam2023one} that this problem requires $\Omega(\sqrt{n})$ communication in the $\Rtwo$ model. We will revisit their $\Qent$ upper bound of $O(\log n)$. 

Alice and Bob share $\log n +1$ EPR pairs. Alice applies $\begin{bmatrix}
A & 0 \\ 0 & C
\end{bmatrix}$ to her part of the state and Bob applies 
$\begin{bmatrix}
B^\dagger & 0 \\ 0 & D^\dagger
\end{bmatrix}$ to his part and they send all their qubits to the referee. 
The referee first does a CNOT on the first two qubits, applies a controlled swap operator between the last two sets of registers controlled on the first register and finally measures the first qubit in the Hadamard basis -- this is depicted in~\Cref{fig:3}. It was shown in~\cite{arunachalam2023one} that the probability with which the referee outputs 1 is at least $0.95$ if $\Tr(ABCD)\ge 0.9n$ and at most $0.55$ if $\Tr(ABCD)\le 0.1n$ and hence, this protocol solves the ABCD problem.

We will now show how to implement the referee's actions in constant $T$-depth. The first gate is CNOT, a Clifford. Before applying each subsequent CSWAP, the referee copies the control qubit onto $\log n$ different ancillary qubits in  the $\ket{0}$ state using CNOTs (Clifford operations). This allows her to then implement all the CSWAPs in parallel. Finally, each CSWAP can be implemented in constant Toffoli depth as depicted in~\Cref{fig:4}, and a Toffoli gate can be implemented with $T$-depth 1~\cite{Tdepthone}. 
Altogether, we obtain a circuit for the referee's actions that acts on $O(\log n)$ qubits and has $T$-depth 1, as desired.  

\begin{figure}
\centering
\begin{subfigure}{0.49\textwidth}
\centering
\scalebox{1.3}{
\mbox{
\Qcircuit @C=1em @R=1em  {
& \targ& \qw & \qw & \qw & \qw \\ 
& \ctrl{-1} & \ctrl{1} & \ctrl{2} & \ctrl{3} & \gate{H} & \meter & \qw \\
& \qw & \qswap & \qw & \qw & \qw\\
& \qw & \qw & \qswap  &\qw & \qw \\
& \qw & \qw & \qw &\qswap &\qw \\
& \qw & \qswap \qw \qwx[-3] & \qw & \qw & \qw \\
& \qw & \qw   &\qswap \qwx[-3] & \qw & \qw \\
& \qw & \qw    & \qw &\qswap \qwx[-3] & \qw \\
}}}
\caption{}\label{fig:3}
\end{subfigure}
\begin{subfigure}{0.49\textwidth}
\centering
\scalebox{1.3}{
\mbox{
\Qcircuit @C=1em @R=1em {
& \ctrl{1} & \qw & \\
& \qswap & \qw & \\
& \qswap \qwx[-1] & \qw  & \\
}
}
\mbox{
\Qcircuit @C=0.4em @R=0.4em {
 &\\
 & \\
= & \\
}}\hspace{1em}
\mbox{
\Qcircuit @C=1em @R=0.4em {
& \qw & \ctrl{1} &\qw & \qw\\
& \targ & \ctrl{1} & \targ &\qw \\
& \ctrl{-1} & \targ &  \ctrl{-1} &\qw \\
}}}
\caption{}\label{fig:4} 
\end{subfigure}
\caption{a) The referee's circuit. b) Implementation of CSWAP using a single Toffoli.}
\end{figure}
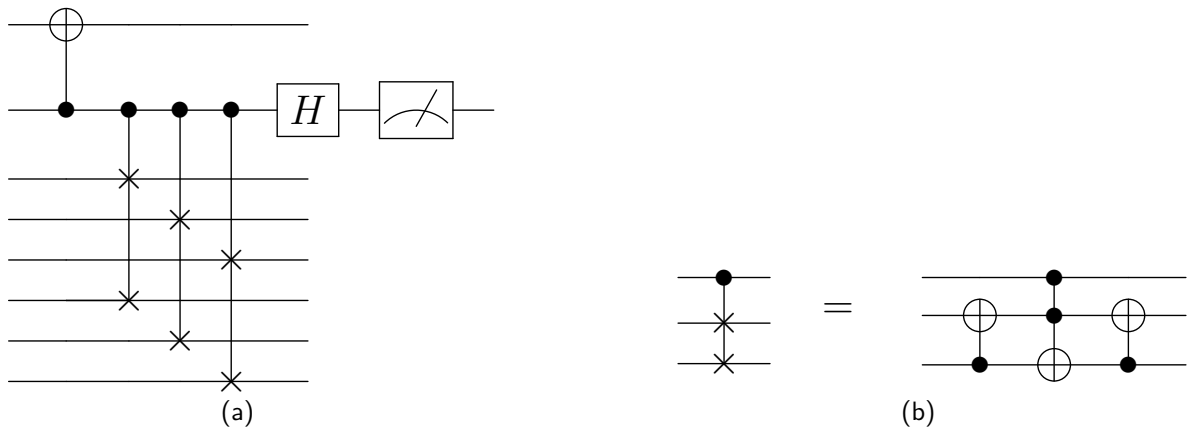

\appendix

\bibliographystyle{unsrtnat}
\bibliography{biblio}

\end{document}